\documentclass[twoside]{article}

\usepackage{amsmath, amsthm, amsfonts, amssymb, mathrsfs, float}
\usepackage{algorithm}
\usepackage{algorithmic}
\usepackage{comment}
\usepackage{bm}
\usepackage{verbatim}
\usepackage{mathtools}
\usepackage{bbm}
\usepackage{dsfont}
\usepackage{booktabs}
\usepackage{relsize}
\usepackage{caption}
\usepackage{subcaption}
\usepackage[english]{babel}
\usepackage{xcolor}
\usepackage[bookmarks=true]{hyperref}
\usepackage{enumitem}
\usepackage{bookmark}

\newtheorem{definition}{Definition}
\newtheorem{assumption}{Assumption}
\newtheorem{lemma}{Lemma}
\newtheorem{proposition}{Proposition}
\newtheorem{theorem}{Theorem}
\newtheorem{corollary}{Corollary}

\DeclareMathOperator{\tr}{tr}

\DeclareMathOperator{\diag}{diag}

\DeclareMathOperator{\sign}{sign}

\newcommand{\st}{\mathrm{s.t.}}

\def\bA{\bm{A}}
\def\tbA{\widetilde{\bm{A}}}
\def\bE{\bm{E}}
\def\bI{\bm{I}}
\def\bU{\bm{U}}
\def\bW{\bm{W}}

\def\bo{\bm{1}}
\def\bu{\bm{u}}
\def\bv{\bm{v}}
\def\bx{\bm{x}}
\def\by{\bm{y}}

\def\mG{\mathcal{G}}

\def\E{\mathbb{E}}
\def\P{\mathbb{P}}
\def\R{\mathbb{R}}

\usepackage[accepted]{aistats2022}
\usepackage[round]{natbib}

\begin{document}

%

%

\twocolumn[

\aistatstitle{Exact Community Recovery over Signed Graphs}

\aistatsauthor{ Xiaolu Wang$^\ast$ \And Peng Wang$^\dagger$\ \And  Anthony Man-Cho So$^\ast$ }

\aistatsaddress{ $^\ast$Department of Systems Engineering and Engineering Management \\ The Chinese University of Hong Kong \\ $^\dagger$Department of Electrical Engineering and Computer Science \\ University of Michigan, Ann Arbor} 
]

\begin{abstract}
Signed graphs encode similarity and dissimilarity relationships among different entities with positive and negative edges. In this paper, we study the problem of community recovery over signed graphs generated by the signed stochastic block model (SSBM) with two equal-sized communities. Our approach is based on the maximum likelihood estimation (MLE) of the SSBM. Unlike many existing approaches, our formulation reveals that the positive and negative edges of a signed graph should be treated unequally. We then propose a simple two-stage iterative algorithm for solving the regularized MLE. It is shown that in the logarithmic degree regime, the proposed algorithm can exactly recover the underlying communities in nearly-linear time at the information-theoretic limit. Numerical results on both synthetic and real data are reported to validate and complement our theoretical developments and demonstrate the efficacy of the proposed method.
\end{abstract}

\section{INTRODUCTION}
Signed graph is a type of undirected graph whose edges can take positive or negative values. Unlike unsigned graphs whose edge labels can only be positive (i.e., $+1$), signed graphs employ negative edges (i.e., $-1$) to characterize the dissimilarity or repulsiveness relationships among different entities. Signed graphs are used to model a wide range of real-world objects, such as friendly and antagonistic relations in social psychology \citep{harary1953notion,cartwright1956structural,davis1967clustering}, common and opposite viewpoints in online social networks \citep{yardi2010dynamic}, trust or distrust between users in bitcoin trading platforms \citep{dittrich2020signal}, gene regulatory networks of bacteria \citep{fujita2012functional}, compensatory interaction networks formed during viral evolution \citep{quadeer2018co}, etc. Recent years have witnessed increasing interest in the analysis of signed graphs in a variety of machine learning tasks, such as graph clustering \citep{kunegis2010spectral,chiang2012scalable,mercado2016clustering,cucuringu2019sponge,mercado2019spectral}, node ranking \citep{shahriari2014ranking}, link prediction \citep{chiang2011exploiting,chiang2014prediction,ye2013predicting,le2017troll}, node classification \citep{bosch2018node}, graph embedding \citep{yuan2017sne,wang2017signed,derr2018signed}, etc. 

The study of signed graphs stems from the seminal work by \cite{harary1953notion} on the social balance theory, where positive edges encode friend relationships and negative edges encode enemy relationships. In this theory, a signed graph is said to be balanced if the network it represents follows the patterns: “a friend of my friend is my friend”, “an enemy of my friend is my enemy”, and “an enemy of my enemy is my friend”. Such graphs are inherently endowed with community structures. More specifically, it is shown in \cite{harary1953notion} that a signed graph is balanced if and only if (i) all of its edges are positive, or (ii) its nodes can be partitioned into 2 disjoint communities such that the edges within communities are positive and the edges between communities are negative.

In this paper, we focus on the problem of identifying the hidden communities over signed graphs that are randomly generated by the signed stochastic block model (SSBM). Specifically, given $n$ nodes that are partitioned into two equal-sized communities, a balanced signed graph is formed by randomly connecting the nodes within the same community with positive edges and the nodes across different communities with negative edges. To make the SSBM more realistic, we allow the presence of noise that yields unbalanced signed graphs, i.e., there can be negative edges within communities and positive edges across communities. In the logarithmic sparsity regime of the SSBM, the information-theoretic limit, which is the necessary and sufficient condition for exact community recovery, is established in the literature (see, e.g., \cite{jog2015information,yun2016optimal}). In this work, we develop a simple and scalable algorithm for exact community recovery in the SSBM at the information-theoretic limit.

\subsection{Related Works}
Community detection, which is also known as graph clustering, is one of the most important tasks in network analysis. Motivated by the social balance theory, the rough goal of community detection over signed graphs is to find a partition of its nodes so that there are as many as possible positive edges within communities and negative edges across communities. 

There are a collection of works on \textit{correlation clustering} for clustering complete signed graphs. The problem is NP-hard while suboptimal partition schemes can be obtained in polynomial time through approximation algorithms. For example, the formulation in \cite{bansal2004correlation,giotis2006correlation} equally weighs the number of positive (resp. negative) edges within clusters and the number of negative (resp. positive) edges across clusters. Later, \cite{coleman2008local} considered correlation clustering with 2 communities based on a local search algorithm. More recently, \cite{puleo2015correlation} extended the correlation clustering model to allow for bounds on the sizes of the communities.

Spectral methods are among the most popular clustering methods,  which proceeds by computing the eigenvectors of a matrix associated with the graph, followed by $k$-means clustering. For example, \cite{kunegis2010spectral} studied the 2-way spectral clustering over signed graphs based on the signed Laplacian matrix. Later, \cite{chiang2012scalable} extended the signed graph clustering problem to $k$ communities by minimizing the balanced normalized graph cut, which yields an eigen-problem as well. Recently, \cite{bonchi2019discovering} considered the problem of detecting 2 small conflicting subsets of nodes in a signed graph, which is formulated as the so-called ``discrete eigenvector'' problem. \cite{chu2016finding,tzeng2020discovering} generalized the problem to $k\geq2$ subsets and approximately solved the resulting formulation using a spectral method. 

Stochastic block model (SBM) is a generative model for random graphs that admit community structures. It provides a useful benchmark for validating and comparing different algorithms for community detection \citep{abbe2018community}. Over the past years, great progress has been made on establishing the fundamental limits for detecting communities in various SBMs \citep{mossel2014consistency,abbe2015exact,abbe2015community,yun2016optimal} and on developing computationally tractable algorithms that can recover the hidden communities at the information-theoretic limits \citep{hajek2016achieving,gao2017achieving,amini2018semidefinite,wang2020nearly,wang2021optimal,wang2021non}. In the context of signed graph clustering, \cite{cucuringu2019sponge} proposed a formulation based on a generalized eigen-problem and \cite{mercado2016clustering,mercado2019spectral} proposed approaches based on the matrix power means Laplacian. They provided recovery guarantees for signed graphs based on SBMs. However, their approaches cannot achieve exact recovery at the information-theoretic limit of the SSBM, while ours is able to do so. \cite{yun2016optimal} established the information-theoretic bound for a general type of labeled stochastic block model and proposed a spectral partition algorithm that achieves exact recovery in $\mathcal{O}(n\log^2n)$ time. Similar to the SSBM, the censored block model (CBM) \citep{abbe2014decoding} generates graphs with ternary edge labels $\{\ast,0,1\}$, where $\ast$ encodes the absence of information. Semidefinite programming is employed to recover the ground-truth communities in the CBM, which is solvable in polynomial time but is not scalable.

\subsection{Main Contributions}
In this work, we study the modeling and algorithmic issues for community detection, or more specifically, community recovery, in the SSBM. Our main contributions are summarized as follows:\\
(1) We derive the maximum likelihood estimation (MLE) formulation, which is nonconvex and discrete, for community recovery in the SSBM. This formulation reveals that the positive and negative parts of a signed graph should be treated unequally, which is different from many approaches in the literature (see, e.g., \cite{kunegis2010spectral,mercado2016clustering,mercado2019spectral}).\\
(2) We propose a simple two-stage iterative algorithm to solve the regularized version of the MLE. Specifically, the first stage employs power iterations (PIs) to obtain an approximate solution. The second stage employs generalized power iterations (GPIs) to iteratively refine the scaled approximate solution. 
We prove that the proposed algorithm can exactly recover the underlying communities in $\mathcal{O}(n\log^2n/\log\log n)$ time at the information-theoretic limit. \\
(3) Since the connectivity parameters of the SSBM are generally unknown, we provide a method to estimate the true parameters with non-asymptotic upper bounds on the estimation errors. The estimation method is based on counting the number of edges and triangles in the graph and can be easily implemented. \\
(4) We also conduct experiments on both synthetic and real data. The numerical results support our theoretical developments and demonstrate the efficacy of our proposed approach. 

\subsection{Notation}
We use lower boldface letters, e.g., $\bm{v}$, to denote vectors. Given a vector $\bv \in \R^n$, we use $v_i$ to denote its $i$-th elements, $\|\bm{v}\|$ to denote its $\ell_2$-norm, and $\bv/|\bv|$ to denote the vector given by
\begin{align*}
	\left( \frac{\bv}{|\bv|} \right)_i = 
	\begin{cases}
		1, &\ \text{if}\ v_i \ge 0,\\
		-1, &\ \text{otherwise},
	\end{cases}\ \text{for}\ i=1,2,\dots,n.
\end{align*} 
We use upper boldface letters, e.g., $\bm{M}$, to denote matrices. Given a matrix $\bm{M}$, we use $M_{ij}$ to denote its $(i,j)$-th element and $\|\bm{M}\|$ to denote its spectral norm. 
Given an integer $m\geq1$, we use $[m]$ to denote the set $\{1,\dots,m\}$.  We use $\mathbf{Bern}(p)$ to denote the Bernoulli random variable with mean $p\in[0,1]$. We use $\bm{I}$, $\bm{E}$, $\bm{0}$, and $\bm{1}$ to denote the identity matrix, the all-one matrix, the all-zero vector, and the all-one vector, respectively, and their dimensions will be clear from the context.  Finally, we use ``w.p.'' as the abbreviation for ``with probability''.

\section{PRELIMINARIES AND MAIN RESULTS}\label{sec:main}
In this section, we first formally introduce the SSBM and derive its corresponding MLE formulation. Then, we present the proposed method for exact community recovery based on a regularized version of the MLE. Finally, we give a summary of the recovery guarantee of our method.

\subsection{Signed Stochastic Block Model}
The SSBM considered in this paper is formally defined as follows: 
\begin{definition}[Signed Stochastic Block Model]\label{def-ssbm}
	Let $n\ge 2$ be an integer, $\mathcal{V}\coloneqq[n]$ be the node set that are partitioned into two equal-sized disjoint communities $\mathcal{V}_1$ and $\mathcal{V}_2$, and $\bm{x}^* \in \{1,-1\}^n$ be the community label vector such that $x^*_i=1$ if $i\in\mathcal{V}_1$ and $x^*_i=-1$ if $i\in\mathcal{V}_2$. Let $p^+,p^-,q^+, q^- \in (0,1]$ be the connectivity probabilities. 
	We say that the signed graph $\mathcal{G}$ is generated according to the signed stochastic block model, denoted by $\textsf{SSBM}(n,\bm{x}^*,p^+,p^-,q^+,q^-)$, if $\mathcal{G}$ has node set $\mathcal{V}$ and the elements $\{A_{ij}\}_{1\le i,j\le n}$ of its adjacency matrix $\bm{{A}}$ are generated independently as follows: \\
	(i) If nodes $i$ and $j$ belong to the same community and $i\neq j$, then
	\begin{align}\label{dist:aij-same}
		A_{ij} = A_{ji} = 
		\begin{cases}
			1, &\text{w.p.\ }\ p^+, \\
			-1,&\text{w.p.\ }\ p^-, \\
			0, &\text{w.p.\ }\ 1-p^+-p^-.
		\end{cases}
	\end{align} 
	(ii) If nodes $i$ and $j$ belong to different communities, then
	\begin{align}\label{dist:aij-diff}
		A_{ij} = A_{ji} = \begin{cases}
			1, &\text{w.p.\ }\ q^+, \\
			-1,&\text{w.p.\ }\ q^-, \\
			0, &\text{w.p.\ }\ 1-q^+-q^-.
		\end{cases}
	\end{align}
	(iii) There are no self-loops in $\mathcal{G}$, i.e, $A_{ii}=0$ for all $i \in [n]$.
\end{definition}
We should mention that there are some other types of signed SBMs in the literature for analyzing community detection methods over signed graphs. In particular, taking $p^+=q^-$ and $p^-=q^+$, $\textsf{SSBM}(n,\bm{x}^*,p^+,p^-,q^+,q^-)$ boils down to the signed SBM in \cite{cucuringu2019sponge} for 2-way clustering. Besides, \cite{mercado2016clustering,mercado2019spectral} considered a slightly different signed SBM that allows simultaneous presence of both positive and negative edges between each pair of nodes.

In this work, our goal is to develop a simple and efficient algorithm that can exactly recover the true community label vector $\bx^*$ based on the adjacency matrix $\bm{A}$ generated by  $\textnormal{\textsf{SSBM}}(n,\bm{x}^*,p^+,p^-,q^+,q^-)$. We are especially interested in the logarithmic degree regime of the SSBM, i.e.,
\begin{align}
	&p^+ = \frac{ \alpha^+\log n}{n},\ p^- = \frac{\alpha^-\log n}{n},\label{p+-}\\
	&q^+ = \beta^+\frac{\log n}{n},\ q^- = \beta^-\frac{\log n}{n},\label{q+-}
\end{align}
where $\alpha^+, \beta^+,\alpha^-,\beta^->0$. The SSBM in this regime exhibits a sharp information-theoretic threshold. Indeed, it is shown in the literature \citep{jog2015information,yun2016optimal} that it is possible to recover $\bm{x}^*$ in $\textnormal{\textsf{SSBM}}\left(n,\bm{x}^*,\frac{ \alpha^+\log n}{n},\frac{ \alpha^+\log n}{n},\frac{ \alpha^+\log n}{n},\frac{ \alpha^+\log n}{n}\right)$ with high probability if and only if
\begin{equation}\label{IT:bound}
	(\sqrt{\alpha^+}-\sqrt{\beta^+})^2 + (\sqrt{\alpha^-}-\sqrt{\beta^-})^2 \ge 2.
\end{equation}
Given the above preliminaries, we make the following assumption that will be used in the subsequent discussion:
\begin{assumption}\label{assump:base}
	The adjacency matrix $\bA$ of the signed graph $\mG$ is generated according to the $\textsf{SSBM}(n,\bm{x}^*,p^+,p^-,q^+,q^-)$ given by Definition \ref{def-ssbm}. Moreover, the connectivity probabilities $p^+,p^-,q^+,q^-$ are given by \eqref{p+-} and \eqref{q+-} with $\alpha^+>\beta^+>0$ and $\beta^->\alpha^->0$.
\end{assumption}
The requirements $\alpha^+>\beta^+$ and $\beta^->\alpha^-$ in Assumption \ref{assump:base} are for ease of technical exposition. Indeed, if $\alpha^+<\beta^+$ or $\beta^-<\alpha^-$, the main theoretical results in this section still remain valid based on similar technical developments.

\subsection{Maximum Likelihood Estimation}
Before we proceed, we let $\mathcal{G}^+$ (resp. $\mathcal{G}^-$) be the subgraph formed by the positive (resp. negative) edges of the signed graph $\mathcal{G}$, whose adjacency matrix $\bA^+$ (resp. $\bA^-$) is given by
\begin{align}\label{eq:aij}
	A^+_{ij} = \max\{A_{ij},0\}\ (\text{resp.\ }A^-_{ij} = \max\{-A_{ij},0\}),
\end{align}
for all $i, j\in[n]$. Obviously, we have $\bm{A}=\bm{A}^+-\bm{A}^-$. By Definition \ref{def-ssbm} and \eqref{eq:aij}, $A_{ij}^+$ and $A_{ij}^-$ are dependent random variables for all $i, j\in [n]$. Moreover, $\{A^+_{ij}\}_{1 \le i < j \le n}$ (resp. $\{A^-_{ij}\}_{1 \le i < j \le n}$) are i.i.d. $\mathbf{Bern}(p^+)$ (resp. $\mathbf{Bern}(p^-)$) if nodes $i$ and $j$ belong to the same community and $\{A^+_{ij}\}_{1 \le i < j \le n}$ (resp. $\{A^-_{ij}\}_{1 \le i < j \le n}$) are $\mathbf{Bern}(q^+)$ (resp. $\mathbf{Bern}(q^-)$) otherwise. Then, we are ready to present the MLE formulation for community recovery in $\textsf{SSBM}(n,\bm{x}^*,p^+,p^-,q^+,q^-)$. 

\begin{proposition}\label{thm:MLE}
	Suppose that Assumption \ref{assump:base} holds. Then, the maximum likelihood estimator of the ground-truth $\bm{x}^*$ is the solution to the following problem:
	\begin{equation}\label{MLE}
		\begin{split}
			\max &\quad \bx^\top(\mu_n\bA^+-\nu_n\bA^-)\bx\\ 
			\st &\quad \bx \in \{1,-1\}^n,\ \bo^\top\bx=0,
		\end{split}
	\end{equation}
	where
	\begin{align*}
		& \mu_n \coloneqq \log\left(\frac{\alpha^+}{\beta^+}\right) + \log\left(\frac{n-(\beta^++\beta^-)\log n}{n- (\alpha^++\alpha^-)\log n}\right),\\
		& \nu_n \coloneqq \log\left(\frac{\beta^-}{\alpha^-}\right) + \log\left(\frac{n- (\alpha^++\alpha^-)\log n }{n-(\beta^++\beta^-)\log n}\right).
	\end{align*}
\end{proposition}
The proof of Proposition \ref{prop-prop1} is deferred to Section \ref{proof-prop1} in the supplementary material, from which we note that the MLE formulation \eqref{MLE} seeks a partition of the nodes so that the signed graph $\mathcal{G}$ is as balance as possible, i.e., there are as many as possible positive edges within communities and negative edges across communities. The weight parameters $\mu_n$ and $\nu_n$, which depend on the connectivity probabilities, indicate that the positive and negative edges actually contribute unequally to the community recovery task except when $\mu_n = \nu_n$. This is different from many signed graph clustering methods that equally treat the positive and negative edges (see, e.g., \cite{giotis2006correlation,coleman2008local,kunegis2010spectral,mercado2016clustering,mercado2019spectral}).

\subsection{Regularized MLE and Parameter Estimation}
Motivated by the MLE formulation in Proposition \ref{thm:MLE} and the community recovery method for unsigned graphs \citep{wang2020nearly}, our approach to community recovery over signed graphs is based on the regularized version of Problem \eqref{MLE}. Upon defining
\begin{align}\label{xi}
	\xi \coloneqq \lim_{n\rightarrow\infty}\frac{\nu_n}{\mu_n} =\frac{\log(\beta^-/\alpha^-)}{\log(\alpha^+/\beta^+)},\ \widetilde{\bm{A}} \coloneqq \bm{A}^+-\xi\bm{A}^-,
\end{align}
we consider the following regularized MLE problem by penalizing the linear constraint $\bo^\top\bx=0$ in the objective function:
\begin{align}\label{RMLE}
	\max\left\{ \bx^\top \bW \bx:\ \bx \in \{1,-1\}^n \right\},
\end{align}
where $\bm{W} = \widetilde{\bm{A}}-\rho\bm{E}$ and $\rho= \bo^\top\tbA\bo/n^2$.
Since $\bm{W}$ is not necessarily positive semi-definite, the objective function of Problem \eqref{RMLE} is in general nonconvex. We should point out that the matrix $\bW$ relies on the weight parameter $\xi$, which is possibly unknown. In the context of the SSBM, we may only observe the graph $\mathcal{G}$, while the connectivity parameters $\alpha^+$, $\beta^+$, $\alpha^-$, $\beta^-$ that defines $\xi$ are unknown. To address this issue, we propose a method to estimate $\xi$ based on $\mG$ with a non-asymptotic upper bound on the estimation error. We formally state it in the following proposition.

\begin{proposition}\label{prop:esti-p-q}
	Suppose that Assumption \ref{assump:base} holds. Let $N^+$ (resp. $N^-$) be the number of edges, $T^+$ (resp. $T^-$) be the number of triangles in $\mG^+$ (resp. $\mG^-$).
	We let 
	\begin{subequations}
		\begin{align}
			\hat{\alpha}^+ &= \frac{1}{\log n}\left(\frac{2{N}^+}{n} + \sqrt[3]{6{T}^+ - \frac{8{N}^{+^3}}{n^3}}\right),\label{eq:alpha^+_hat} \\
			\hat{\beta}^+ &= \frac{1}{\log n}\left(\frac{2{N}^+}{n} - \sqrt[3]{6{T}^+ - \frac{8{N}^{+^3}}{n^3}}\right),\label{eq:beta^+_hat} \\
			\hat{\alpha}^- &= \frac{1}{\log n}\left(\frac{2{N}^-}{n} + \sqrt[3]{6{T}^- - \frac{8{N}^{-^3}}{n^3}}\right),\label{eq:alpha^-_hat} \\
			\hat{\beta}^- &= \frac{1}{\log n}\left(\frac{2{N}^-}{n} - \sqrt[3]{6{T}^- - \frac{8{N}^{-^3}}{n^3}}\right),\label{eq:beta^-_hat}
		\end{align}
	\end{subequations}
	and
	\begin{align}\label{eq:xi-hat}
		\hat{\xi} \coloneqq \frac{\log(\hat{\beta}^-/\hat{\alpha}^-)}{\log(\hat{\alpha}^+/\hat{\beta}^+)}.
	\end{align}
	Then, for sufficiently large $n$, it holds with probability at least $1-n^{-\Omega(1)}$ that
	\begin{align*}
		&  | \hat{\xi} - \xi | \le \kappa\left(\log n \right)^{-\frac{1}{12}},
	\end{align*}
	where $\kappa > 0$ is a constant.
\end{proposition}
The above parameter estimation method requires knowledge of the numbers of edges and triangles in $\mG^+$ and $\mG^-$, respectively. Indeed, given the adjacency matrix $\bA^+$ (resp. $\bA^-$), we can obtain the number of edges in $\mG^+$ (resp. $\mG^-$) by the formula
\begin{align*}
	N^+ = \sum_{1 \le i < j \le n}A_{ij}^+\ (\text{resp.\ } N^- = \sum_{1 \le i < j \le n}A_{ij}^-)
\end{align*}
and the number of triangles in $\mG^+$ (resp. $\mG^-$) by the formula
\begin{align*}
	T^+ = \frac{1}{6}\tr\left((\bA^+)^3\right)\ (\text{resp.\ } T^- = \frac{1}{6}\tr\left((\bA^-)^3\right)).
\end{align*}
Armed with Proposition \ref{prop:esti-p-q}, if the true connectivity parameters $\alpha^+$, $\beta^+$, $\alpha^-$, $\beta^-$ are unknown, the weight parameter $\xi$ in Problem \eqref{RMLE} can be replaced with its consistent estimator $\hat{\xi}$.

\subsection{Optimization Algorithm and Its Recovery Guarantee}\label{subsec:opt}
We observe that Problem \eqref{RMLE} is analogous to the principal component analysis problem. It is thus natural to apply a power iteration-type method to solve it, which is described in Algorithm \ref{alg-pmgpm}. Specifically, Algorithm \ref{alg-pmgpm} consists of two stages. In the first stage (lines 5--8), it uses the PIs to approximate the leading eigenvector of $\bm{W}$. In the second stage (lines 9--12), it employs the GPIs to iteratively refine the approximate eigenvector after proper scaling. In line 11, $\bm{W}\bm{x}^{t-1}/{|\bm{W}\bm{x}^{t-1}|}$ essentially computes the projection of $\bW\bx^{t-1}$ onto the set $\{1,-1\}^n$. Algorithm \ref{alg-pmgpm} is a generalization of the iterative method in \cite{wang2020nearly} for community recovery over unsigned graphs. Indeed, by setting $\xi=0$ in line 2, Algorithm \ref{alg-pmgpm} reduces to the algorithm in \cite{wang2020nearly}.

\begin{algorithm}[t]
	\caption{Algorithm for Solving Problem \eqref{RMLE}}
	\begin{algorithmic}[1]\label{alg-pmgpm}
		\STATE \textbf{Input:} Signed graph $\mG$ with adjacency matrices $\bm{A}^+,\bm{A}^-$, positive integers $T_1,T_2$
		\STATE Set ${\xi} \leftarrow\log\left({{\beta}^-}/{{\alpha}^-}\right) \big/ {\log\left({{\alpha}^+}/{{\beta}^+}\right)}$ 
		\STATE Set $\widetilde{\bm{A}} \leftarrow \bm{A}^+-{\xi}\bm{A}^-$ and $\rho \leftarrow \bo^\top\tbA\bo/n^2$
		\STATE Set $\bm{W}\leftarrow\widetilde{\bm{A}}-\rho\bm{E}$ 
		\STATE Generate $\bm{y}^{0}$ uniformly distributed over the unit sphere 
		\FOR{$t=1,2,\dots,T_1$}
		\STATE Set $\bm{y}^{t} \leftarrow \bm{W}\bm{y}^{t-1}/{\|\bm{W}\bm{y}^{t-1}\|}$
		\ENDFOR
		\STATE Set $\bm{x}^{0} \leftarrow \sqrt{n}\bm{y}^{T_1}$
		\FOR{$t=1,2,\dots,T_2$}
		\STATE Set $\bm{x}^{t}\leftarrow \bm{W}\bm{x}^{t-1}/{|\bm{W}\bm{x}^{t-1}|}$
		\ENDFOR
	\end{algorithmic}
\end{algorithm}

Next, we present the main theorem of this paper, which establishes the exact recovery guarantee of Algorithm \ref{alg-pmgpm}. 
\begin{theorem}\label{thm-iter-comp}
	Suppose that Assumption \ref{assump:base} holds and $\alpha^+,\beta^+,\alpha^-,\beta^-$ satisfy \eqref{IT:bound}. Then, for sufficiently large $n$, it holds with probability at least $1-n^{-\Omega(1)}$ that Algorithm \ref{alg-pmgpm} outputs $\bm{x}^*$ or $-\bm{x}^*$ in $T_1=\mathcal{O}(\log n/\log\log n)$ PIs and $T_2=\mathcal{O}(\log n/\log\log n)$ GPIs.
\end{theorem}
Theorem \ref{thm-iter-comp} indicates that under the information-theoretic limit \eqref{IT:bound}, Algorithm \ref{alg-pmgpm} is able to exactly recover the ground-truth $\bx^*$ up to a sign with high probability. Equipped with the iteration complexity given in Theorem \ref{thm:MLE}, we can further obtain the time complexity of Algorithm \ref{alg-pmgpm}.
\begin{corollary}\label{coro-coro2}
	Suppose that Assumption \ref{assump:base} holds and $\alpha^+,\beta^+,\alpha^-,\beta^-$ satisfy \eqref{IT:bound}. Then, for sufficiently large $n$, it holds with probability at least $1-n^{-\Omega(1)}$ that Algorithm \ref{alg-pmgpm} outputs $\bm{x}^*$ or $-\bm{x}^*$ in $\mathcal{O}(n\log^2 n/\log\log n)$ time.
\end{corollary}
Corollary \ref{coro-coro2} indicates that Algorithm \ref{alg-pmgpm} achieves exact recovery in nearly linear time. 

It should be pointed out that Algorithm \ref{alg-pmgpm} implicitly requires knowledge of the connectivity parameters $\alpha^+,\beta^+,\alpha^-,\beta^-$. Nevertheless, if these parameters are unknown, we can replace $\xi$ with its consistent estimator $\hat{\xi}$ given by \eqref{eq:xi-hat}. Proposition \ref{prop:esti-p-q} guarantees that the true $\xi$ can be accurately estimated for sufficiently large $n$, thus the results in Theorem \ref{thm-iter-comp} and Corollary \ref{coro-coro2} still remain valid. 

\section{PROOFS OF MAIN RESULTS}\label{sec:proofs}
In this section, we provide the major ingredients in the proofs of the results in Section \ref{sec:main}. This involves estimating the connectivity parameters in the SSBM in Section \ref{subsec:esti-para} and analyzing the behavior of the PIs and GPIs in Sections \ref{subsec:PM} and \ref{subsec:GPM}, respectively. The detailed proofs can be found in the supplementary material. 

\subsection{Estimation of the Connectivity Parameters}\label{subsec:esti-para}
We first present the following lemma about the concentration properties of the numbers of edges and triangles in $\mG^+$ and $\mG^-$, respectively.
\begin{lemma}\label{lem:edges-triangles}
Consider the same setting as Proposition \ref{prop:esti-p-q}. Suppose that $n \ge \max\left\{\alpha^+,\beta^- \right\}$. Then, it holds with probability at least $1-n^{-\Omega(1)}$ that
\begin{align*}
	& \left| N^+  - a\left(\alpha^+ + \beta^+\right) \right| \le \left(\sqrt{n}+\alpha^+/2\right)\log n,\\ 
	& \left| N^-  - a\left(\alpha^- + \beta^-\right) \right| \le \left(\sqrt{n}+\alpha^-/2\right)\log n,\\ 
	& \left| T^+  - b\left( (\alpha^{+})^3 + 3\alpha^+(\beta^+)^2 \right)  \right|  \le 8(\alpha^+)^{\frac{9}{4}}\left(\log n\right)^{\frac{11}{4}},\\
	&  \left| T^-  - b\left( (\alpha^{-})^3 + 3\alpha^-(\beta^-)^2 \right)  \right| \le  8(\beta^-)^{\frac{9}{4}}\left(\log n\right)^{\frac{11}{4}},
\end{align*} 
where
\begin{align}\label{eq:abcd}
	a \coloneqq \frac{n\log n}{4},\ b \coloneqq \frac{\log^3n}{24}.
\end{align}
\end{lemma}
Based on Lemma \ref{lem:edges-triangles}, we can obtain the non-asymptotic guarantee for the estimation of connectivity parameters.
\begin{lemma}\label{lem:unique-sol}
Consider the same setting as Proposition \ref{prop:esti-p-q}. Suppose that $n \ge \max\left\{\alpha^+,\beta^- \right\}$. Then, it holds with probability at least $1-n^{-\Omega(1)}$ that
\begin{align}
	& \max\left\{ | \alpha^+ - \hat{\alpha}^+|, | \beta^+ - \hat{\beta}^+|\right\}  \le \kappa_1\left(\log n \right)^{-\frac{1}{12}},\label{rst1:lem:unique-sol}\\
	& \max\left\{ | \alpha^- - \hat{\alpha}^-|, | \beta^- - \hat{\beta}^-|\right\}  \le \kappa_2\left(\log n \right)^{-\frac{1}{12}}, \label{rst2:lem:unique-sol}
\end{align}
where $\kappa_1,\kappa_2 > 0$ are constants. 
\end{lemma}
Equipped with Lemma \ref{lem:unique-sol}, we can show Proposition \ref{prop:esti-p-q}, whose proof is deferred to Section \ref{proof-prop2} in the supplementary material.

\subsection{Analysis of the Power Iterations}\label{subsec:PM} 
In our technical developments in Sections \ref{subsec:PM} and \ref{subsec:GPM}, we use the true parameters $\alpha^+,\beta^+,\alpha^-,\beta^-$ for conciseness of exposition. As discussed in the last paragraph of Section \ref{subsec:opt}, this does not influence our main results. In this subsection, we analyze the performance of the PIs. We first present the spectral property of the matrix $\bm{W}$.
\begin{lemma}\label{lem:eig-W}
Suppose that Assumption \ref{assump:base} holds.
Let $\lambda_1 \ge \lambda_2 \ge \dots\geq\lambda_n$ be the eigenvalues of $\bm{W}$ and $\bu_1$ be the eigenvector associated with eigenvalue $\lambda_1$. Then, for all sufficiently large $n$, it holds with probability at least $1-n^{-\Omega(1)}$ that
\begin{align*}
	 |\lambda_1| &\geq  \frac{c_0}{2}\log n - 3c_1\sqrt{\log n}, \\ 
	|\lambda_i| &\leq 3c_1\sqrt{\log n},\ \text{for}\ i=2,\dots,n,
\end{align*}
and
\begin{align}\label{eq:u1-x*}
	\min_{ \theta \in \{\pm 1\}} \left\| \bu_1 - \theta\frac{\bx^*}{\sqrt{n}} \right\| \le \frac{c_2}{\sqrt{\log n}},
\end{align}
where $c_0=\left|(\alpha^+-\beta^+)-\xi(\alpha^--\beta^-)\right|$ and $c_1,c_2 > 0$ are constants.
\end{lemma}
Lemma \ref{lem:eig-W} guarantees that the magnitude of the leading eigenvalue of $\bm{W}$ is sufficiently larger than other eigenvalues, which can be used to show that the leading eigenvector $\bm{u}_1$ is sufficiently close to the scaled ground-truth $\bm{x}^*/\sqrt{n}$ up to a sign. Further, we show in the next proposition that the iterates $\bm{y}^{t}$ in the first stage of Algorithm \ref{alg-pmgpm} converge to $\bm{u}_1$ at a linear rate. 
\begin{proposition}\label{prop-prop1}
Suppose that Assumption \ref{assump:base} holds. 
Let $\bu_1$ be the leading eigenvector of $\bW$ and $\{\by^t\}_{t\ge 0}$ be the sequence generated in the first stage of Algorithm \ref{alg-pmgpm}. 
Then, for sufficiently large $n$, it holds with probability at least $1-n^{-\Omega(1)}$ that
\begin{align*}
	&\min_{\theta\in\{\pm 1\}}\ \left\|\bm{y}^t - \theta\bm{u}_1\right\| \le n\left( \frac{6c_1}{c_0\sqrt{\log n} - 6c_1} \right)^t,
\end{align*}
for all $t\ge 0$, where $c_0$ and $c_1$ are the constants in Lemma \ref{lem:eig-W}. 
\end{proposition}

\begin{figure*}
\begin{subfigure}[b]{0.180\textwidth}
	\centering
	\includegraphics[width=\textwidth]{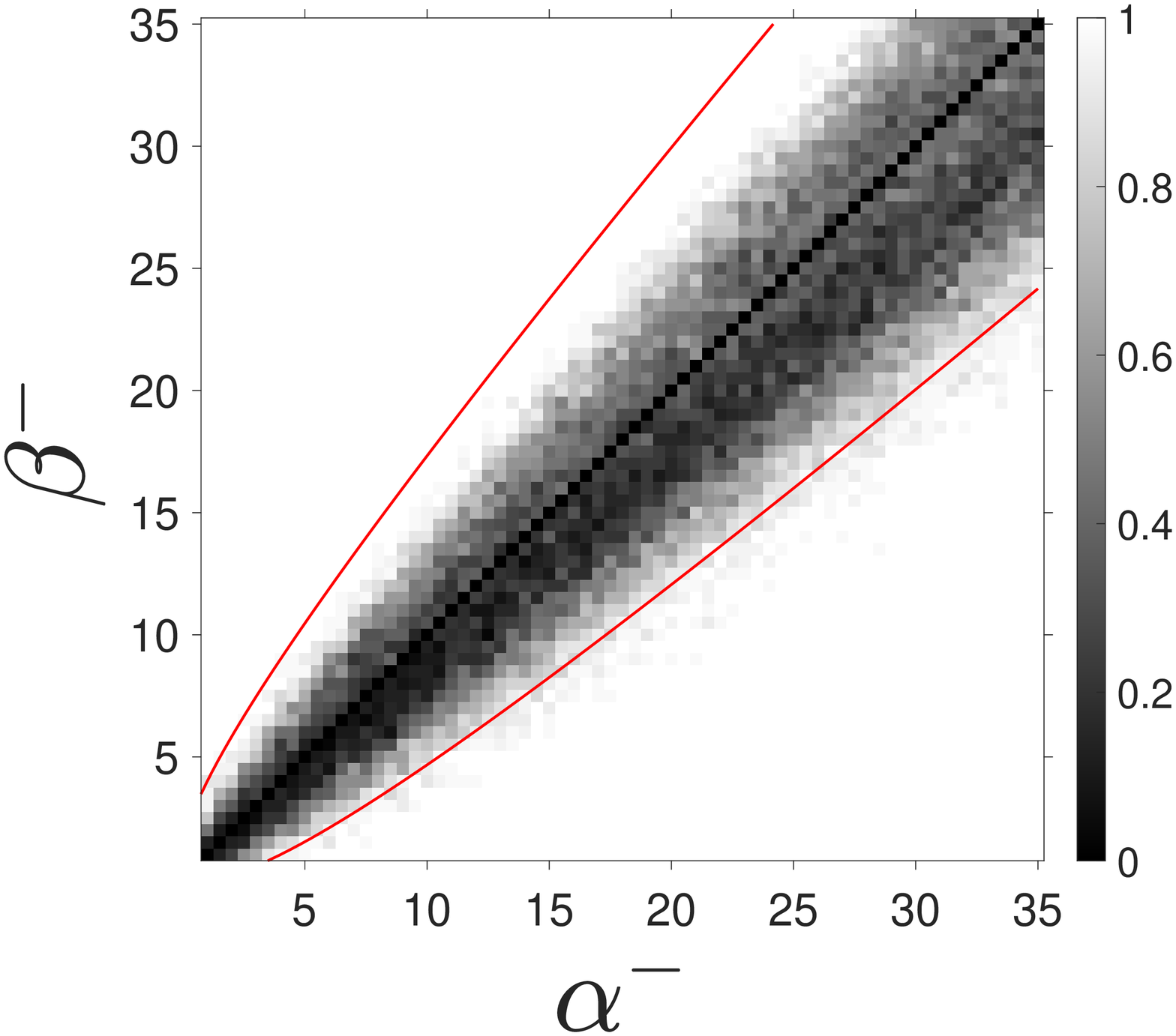}
	\caption{SGPI}
\end{subfigure}
\hfill
\begin{subfigure}[b]{0.180\textwidth}
	\centering
	\includegraphics[width=\textwidth]{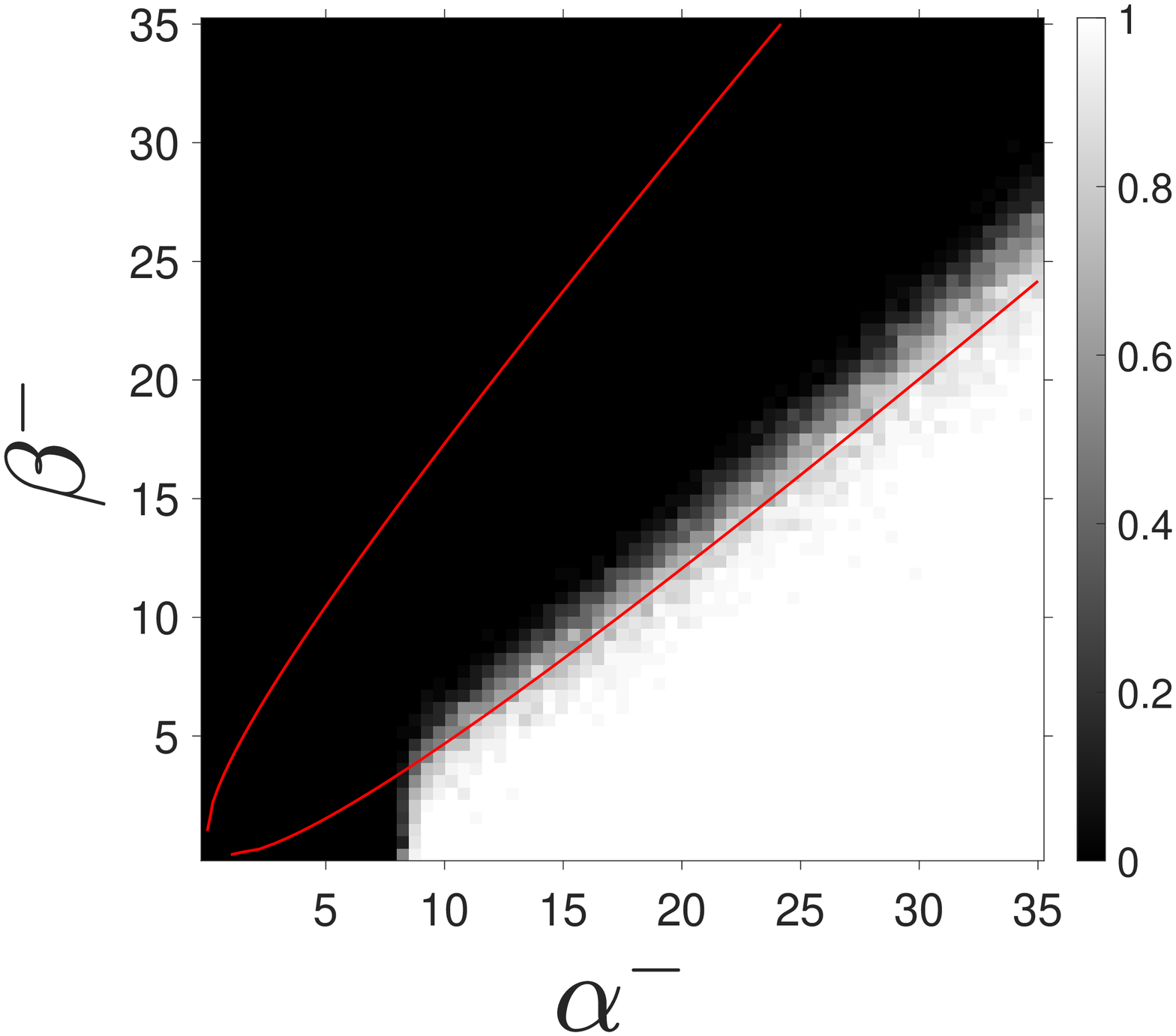}
	\caption{SRC}
\end{subfigure}
\hfill
\begin{subfigure}[b]{0.180\textwidth}
	\centering
	\includegraphics[width=\textwidth]{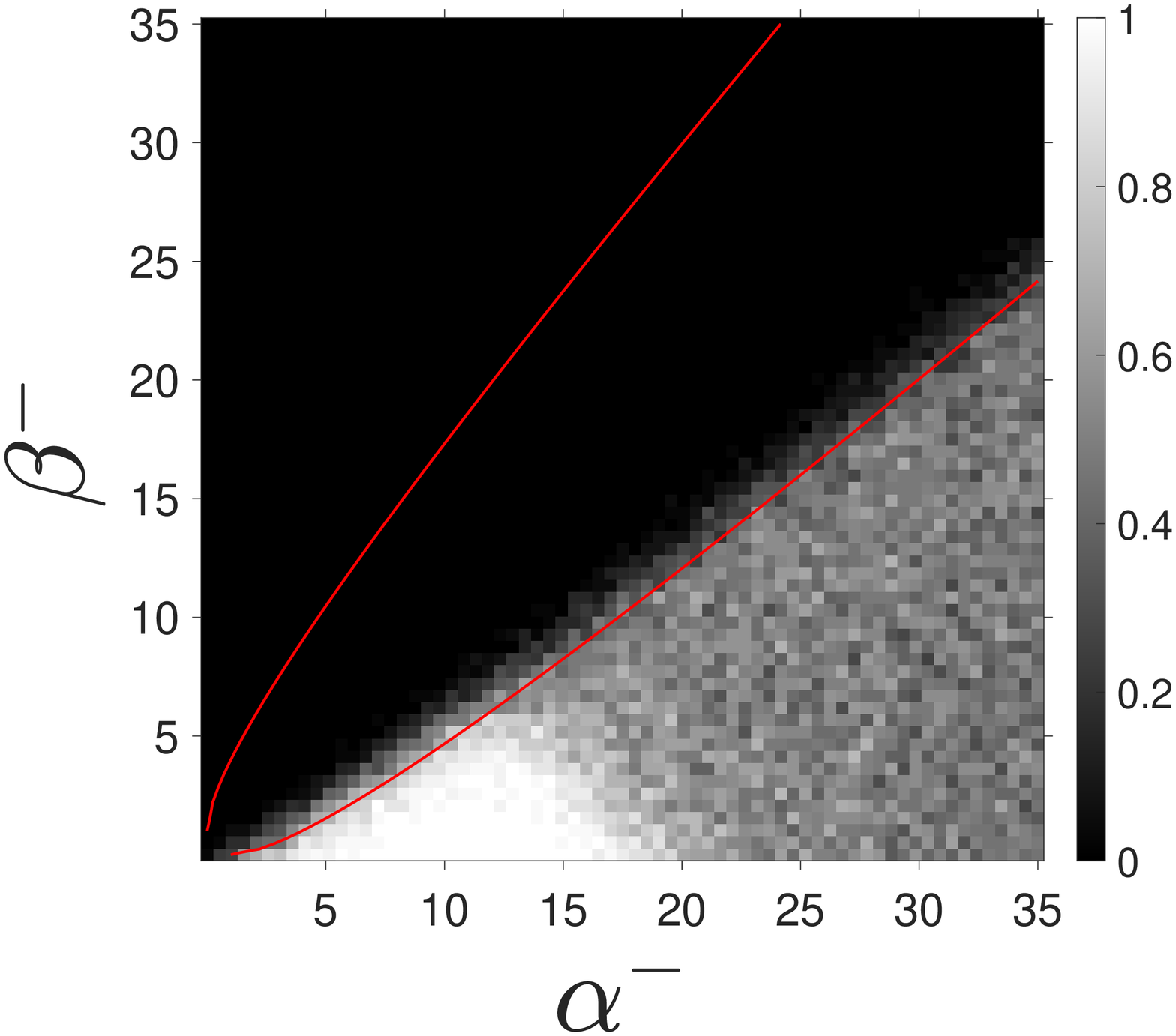}
	\caption{SPONGE}
\end{subfigure}
\hfill
\begin{subfigure}[b]{0.180\textwidth}
	\centering
	\includegraphics[width=\textwidth]{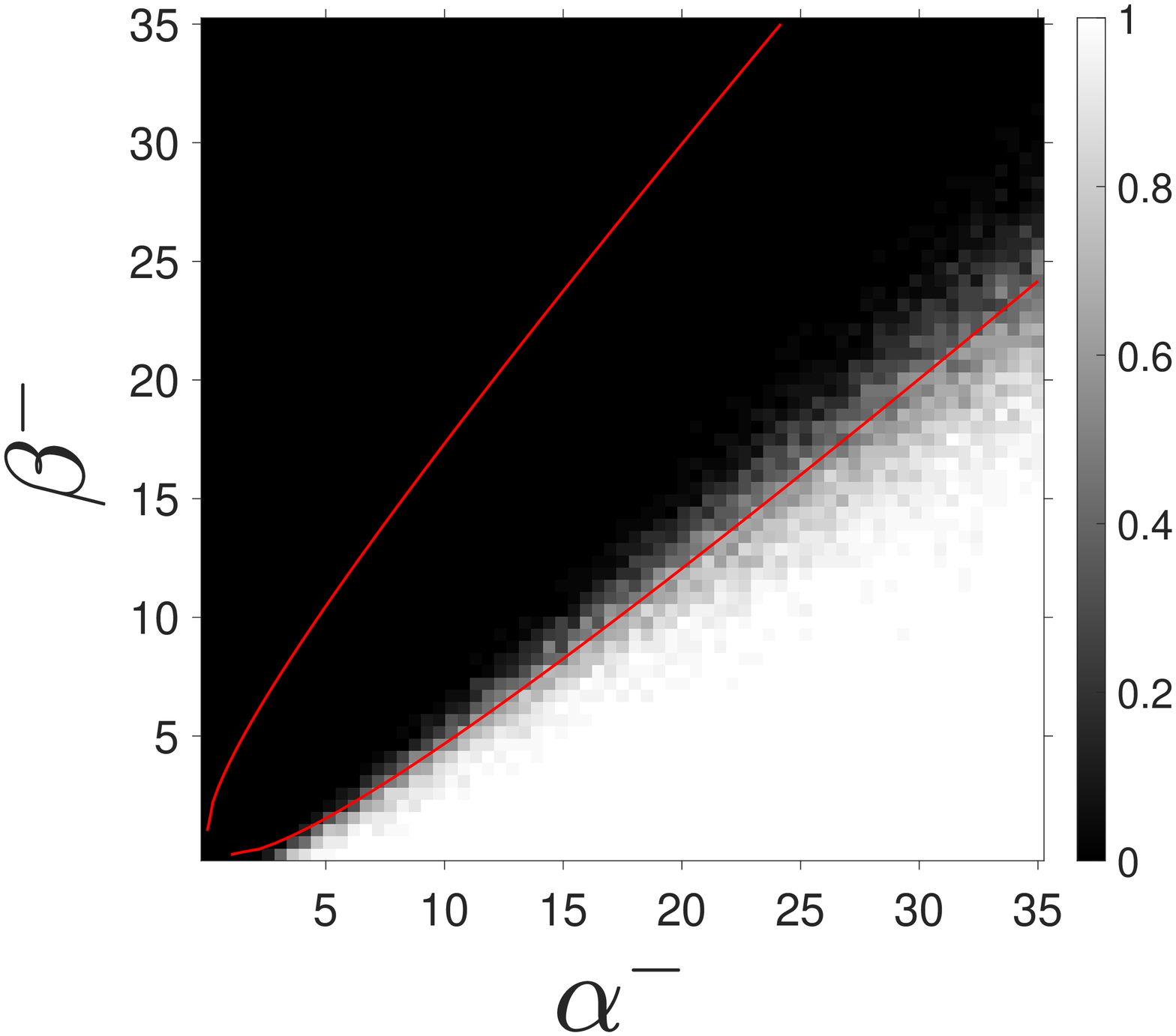}
	\caption{$\text{SPM}_{p=0}$}
\end{subfigure}
\hfill
\begin{subfigure}[b]{0.180\textwidth}
	\centering
	\includegraphics[width=\textwidth]{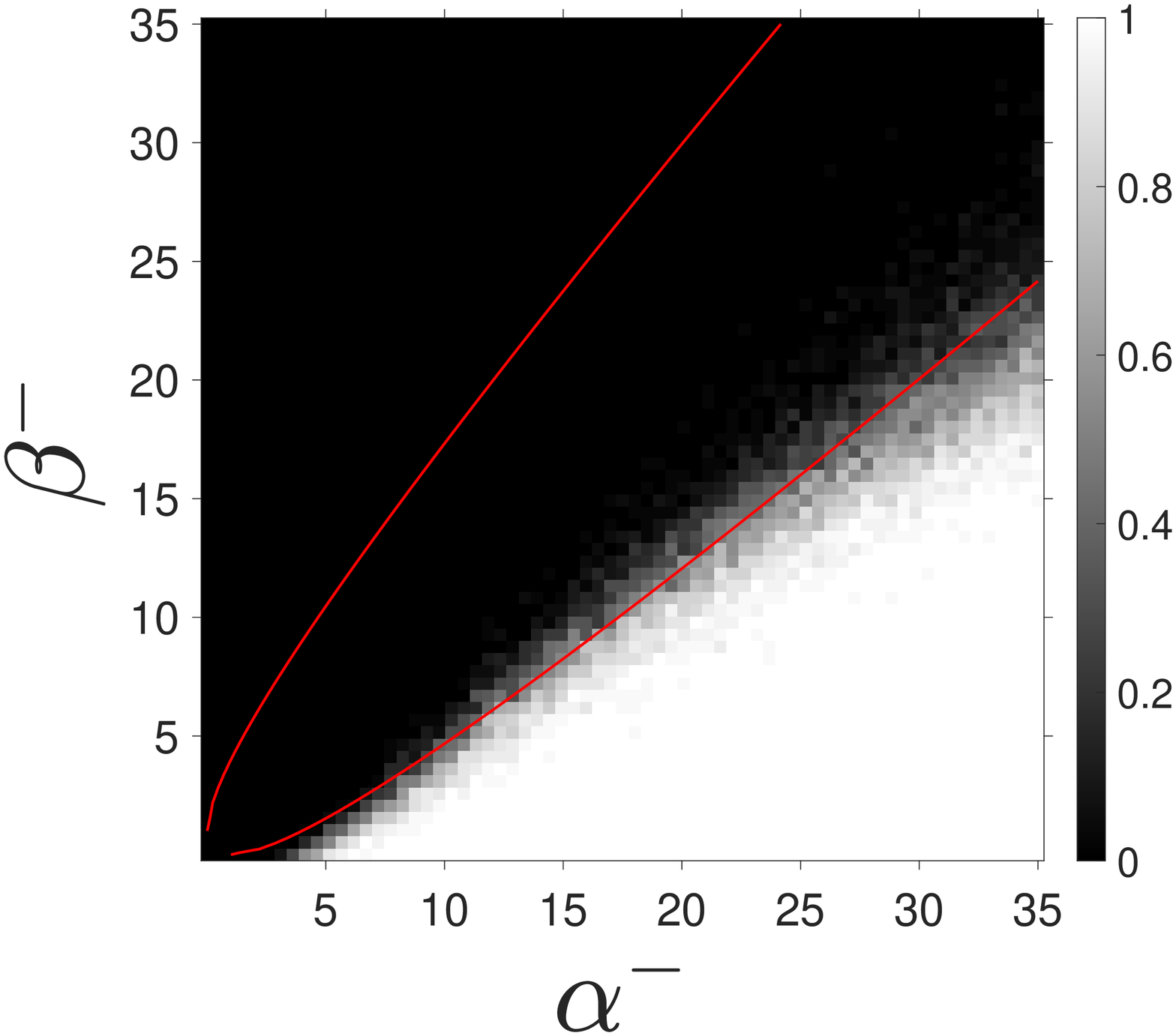}
	\caption{$\text{SPM}_{p=-10}$}
\end{subfigure}
\hfill
\caption{Exact recovery ratios for different values of $\alpha^-$ and $\beta^-$ with an increment of 0.5 (fixing $n=300$, $\alpha^+ = 16$ and $\beta^+=9$). The information-theoretic threshold $(\sqrt{\alpha^-}-\sqrt{\beta^-})^2=1$ is plotted in red.}
\label{fig-alpha-_vs_beta-}
\end{figure*}

\begin{figure*}
\begin{subfigure}[b]{0.180\textwidth}
	\centering
	\includegraphics[width=\textwidth]{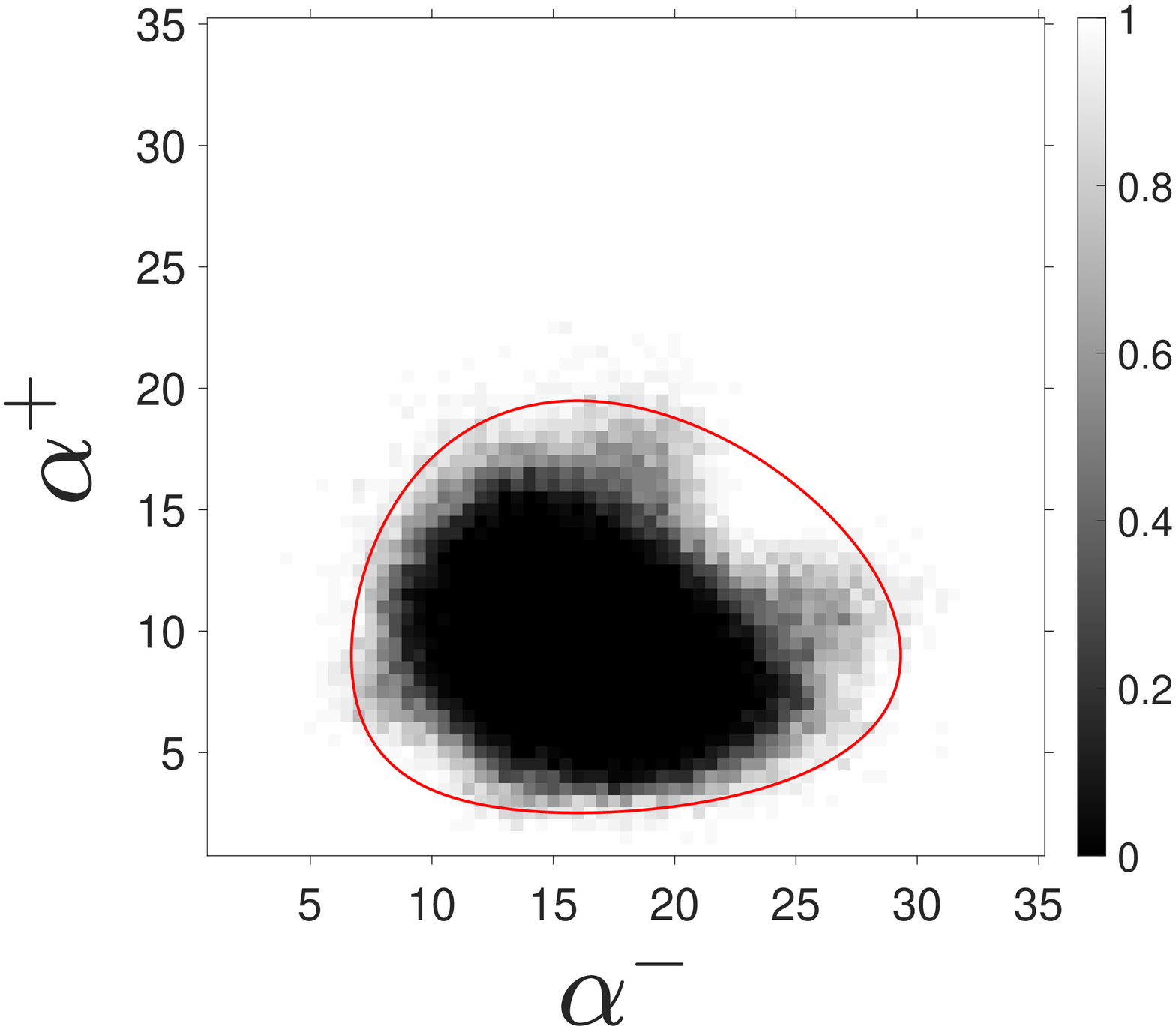}
	\caption{SGPI}
\end{subfigure}
\hfill
\begin{subfigure}[b]{0.180\textwidth}
	\centering
	\includegraphics[width=\textwidth]{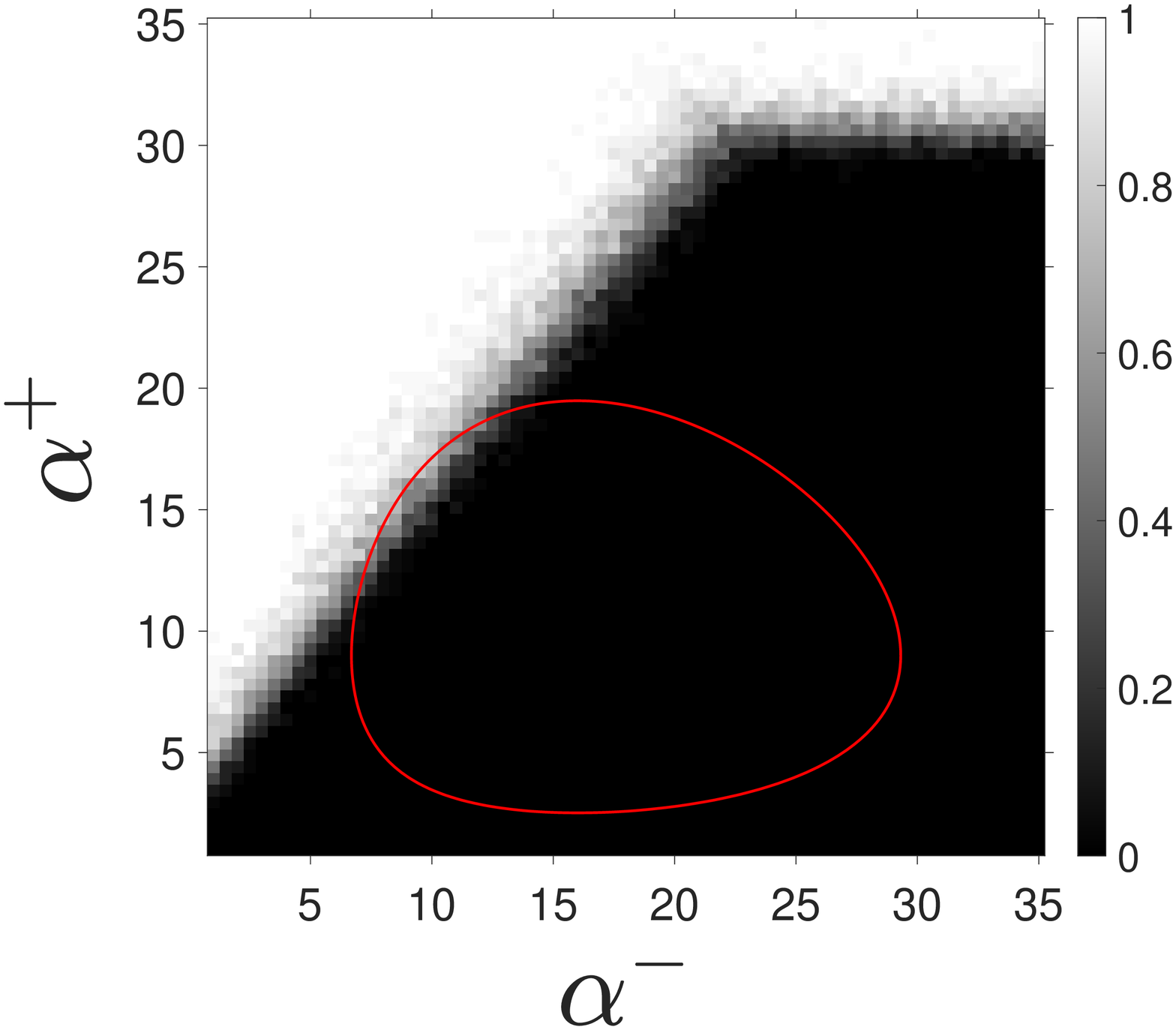}
	\caption{SRC}
\end{subfigure}
\hfill
\begin{subfigure}[b]{0.180\textwidth}
	\centering
	\includegraphics[width=\textwidth]{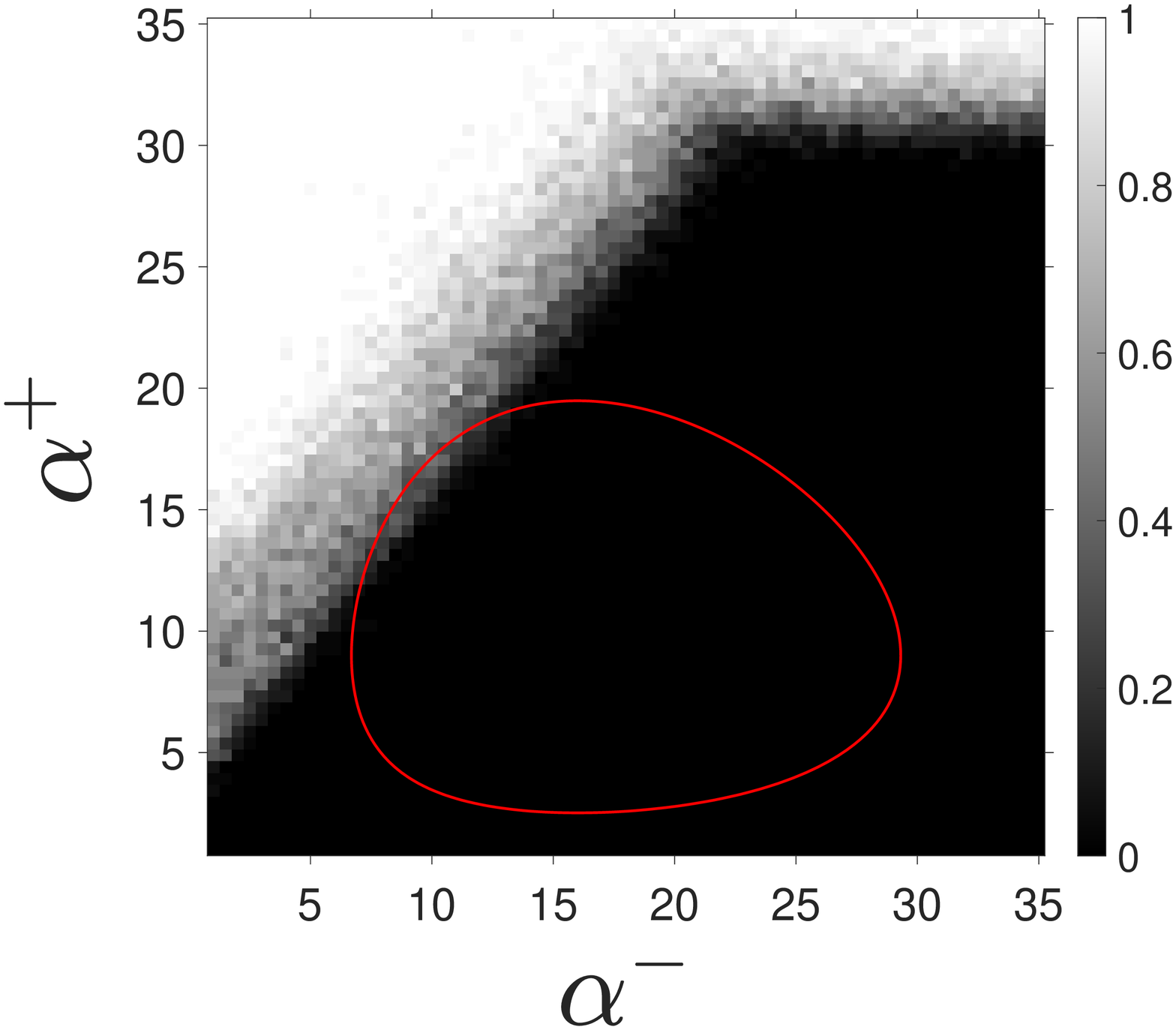}
	\caption{SPONGE}
\end{subfigure}
\hfill
\begin{subfigure}[b]{0.180\textwidth}
	\centering
	\includegraphics[width=\textwidth]{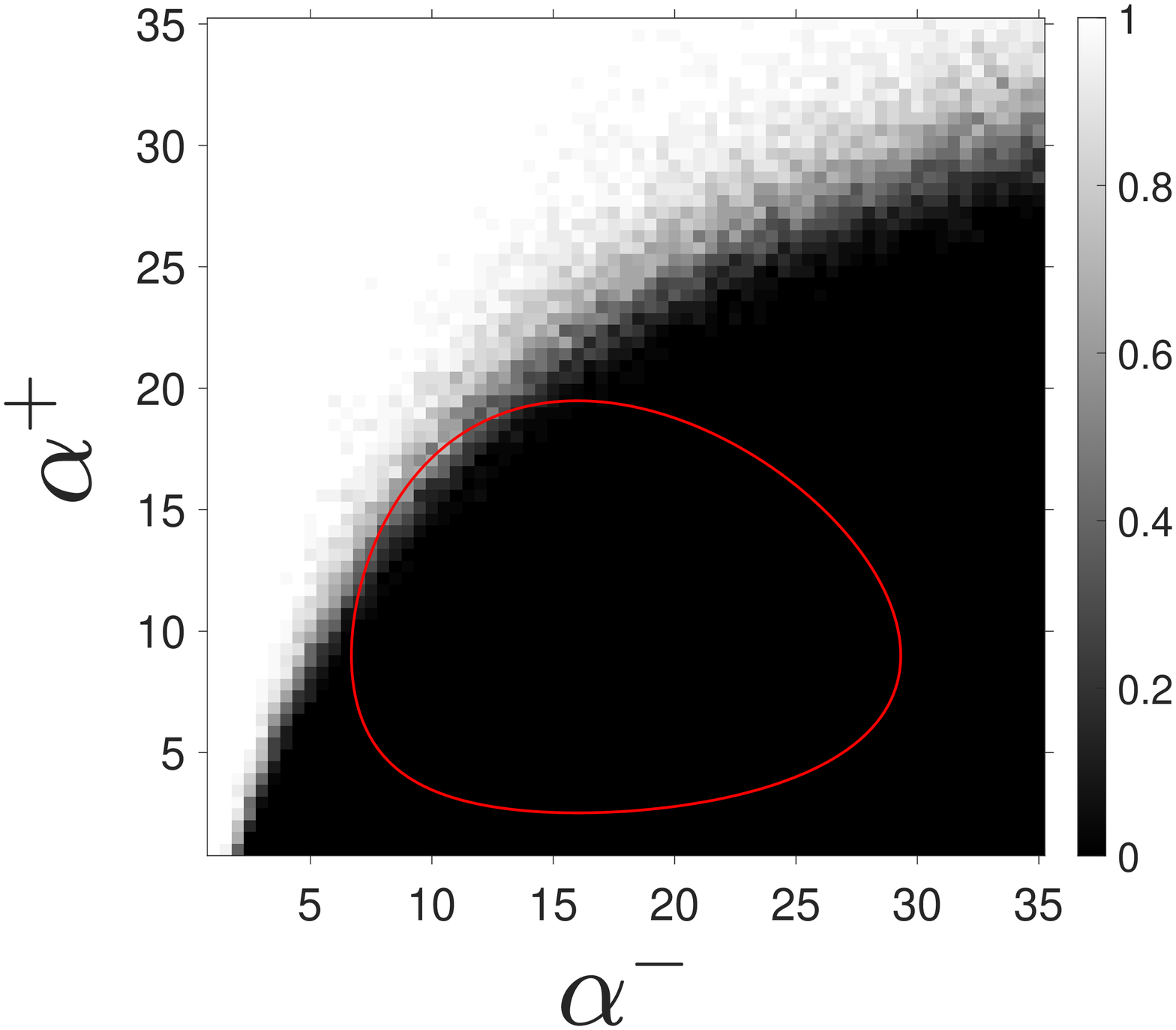}
	\caption{$\text{SPM}_{p=0}$}
\end{subfigure}
\hfill
\begin{subfigure}[b]{0.180\textwidth}
	\centering
	\includegraphics[width=\textwidth]{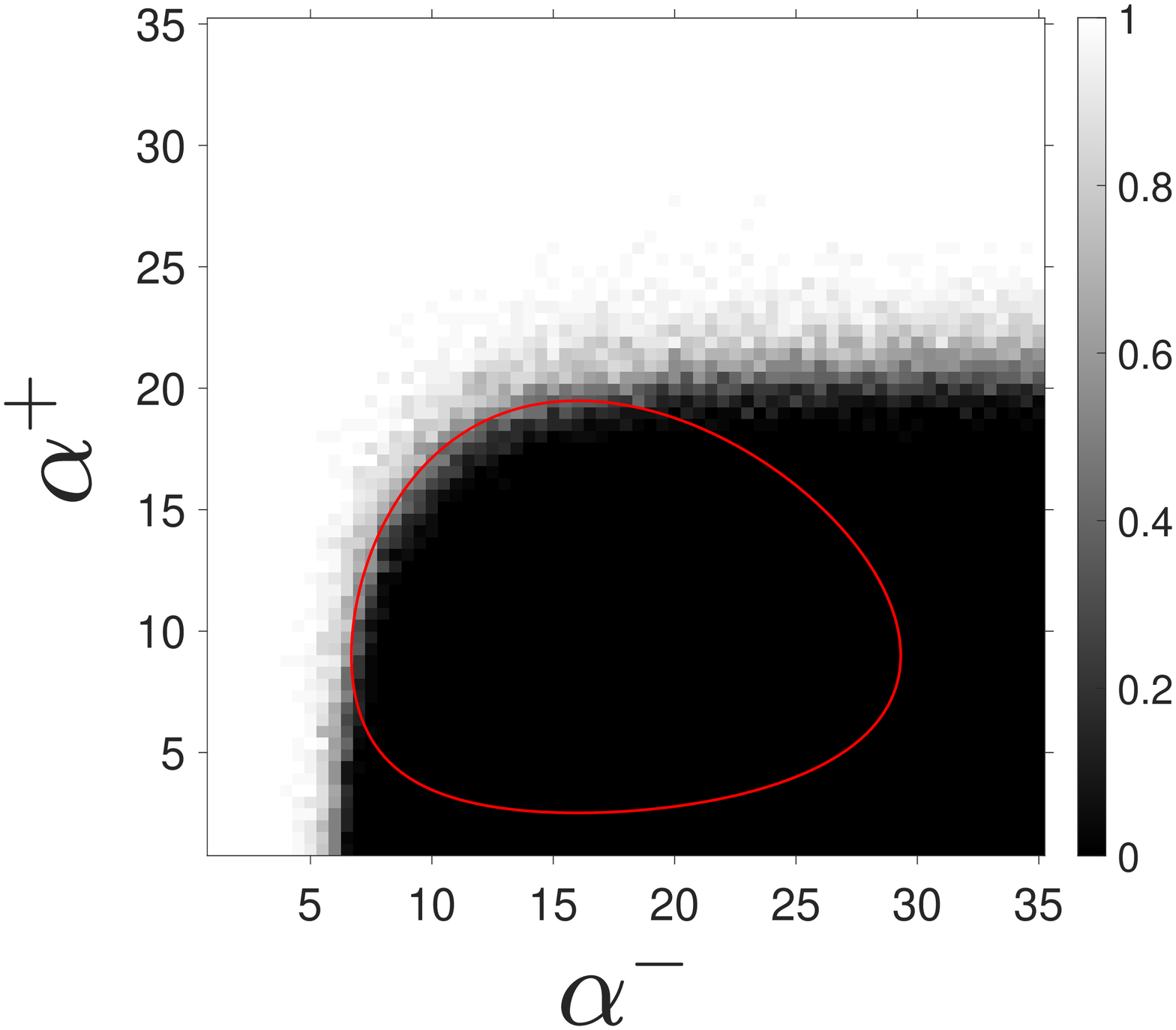}
	\caption{$\text{SPM}_{p=-10}$}
\end{subfigure}
\hfill
\caption{Exact recovery ratios for different values of $\alpha^+$ and $\alpha^-$ with an increment of 0.5 (fixing $n=300$, $\beta^+ = 9$ and $\beta^-=16$). The information-theoretic threshold $(\sqrt{\alpha^+}-3)^2+(\sqrt{\alpha^-}-4)^2=2$ is plotted in red.}
\label{fig-alpha+_vs_alpha-}
\end{figure*}

\subsection{Analysis of the Generalized Power Iterations}\label{subsec:GPM}

In this subsection, we study the convergence behavior of the GPIs. In particular, we establish the contraction property and finite termination property of the GPIs. We begin with a key lemma that says that the magnitude of each entry of $\bW\bx^*$ is at least in the order of $\log n$.

\begin{lemma}\label{lem:Wx}
Suppose that Assumption \ref{assump:base} holds and $\alpha^+,\beta^+,\alpha^-,\beta^-$ satisfy \eqref{IT:bound}. 
Then, for sufficiently large $n$, there exists a constant $\gamma>0$ such that the following inequality holds with probability at least $1-n^{-\Omega(1)}$:
\begin{align}
	&\min\{x_i^*(\bm{W}\bm{x}^*)_i:i=1,\dots,n\} \geq \gamma\log n. \label{rst:lem:Wx}
\end{align}
\end{lemma}
We remark that this lemma cannot be directly implied by \citet[Lemma 8]{abbe2020entrywise}, since it concerns categorical random variables rather than Bernoulli random variables. Moreover, this lemma holds at the information-theoretic limit, which is critical to the optimal recovery performance of Algorithm \ref{alg-pmgpm} and could be of independent interest. Next, we present a lemma that shows the contraction property of the GPIs.
\begin{proposition}\label{prop:contraction}
Suppose that Assumption \ref{assump:base} holds and $\alpha^+,\beta^+,\alpha^-,\beta^-$ satisfy \eqref{IT:bound}.
Let $\bx \in \R^n$ be such that $\|\bm{x}\|=\sqrt{n}$. Then, for sufficiently large $n$, either of the following two statements holds with probability at least $1-n^{-\Omega(1)}$: \\
(i) If $\|\bm{x}-\bm{x}^*\| \leq 2c_2\sqrt{{n}/{\log n}}$,
then we have
\begin{align*}
	\left\| \frac{\bW\bx}{|\bW\bx|}-\bm{x^*}\right\| \le \frac{2c_3}{\gamma\sqrt{\log n}}\|\bm{x}-\bm{x}^*\|.
\end{align*}
(ii) If $\|\bm{x}+\bm{x}^*\| \leq 2c_2\sqrt{{n}/{\log n}}$,
then we have
\begin{align*}
	\left\| \frac{\bW\bx}{|\bW\bx|}+\bm{x^*}\right\| \le \frac{2c_3}{\gamma\sqrt{\log n}}\|\bm{x}+\bm{x}^*\|,
\end{align*}
where $c_3 = (6\sqrt{2}+1)c_1+2|\alpha^+-\xi \alpha^-|+ 1+\xi$, $c_1, c_2$ are the constants in Lemma \ref{lem:eig-W}, and $\gamma>0$ is the constant in Lemma \ref{lem:Wx}.
\end{proposition}
Proposition \ref{prop:contraction} indicates that the iterates $\bm{x}^t$ in the second stage of Algorithm \ref{alg-pmgpm} converge linearly to either $\bm{x}^*$ or $-\bm{x}^*$ within the contraction region with radius $\mathcal{O}(\sqrt{n/\log n})$. Then, the following proposition characterizes a smaller one-step convergence region with radius 2, which further implies the finite termination of the GPIs.
\begin{proposition}\label{prop-one-step-conv}
Suppose that Assumption \ref{assump:base} holds and $\alpha^+,\beta^+,\alpha^-,\beta^-$ satisfy \eqref{IT:bound}. Let $\bx\in \{1,-1\}^n$. Then, for sufficiently large $n$, either of the following two statements holds with probability at least $1-n^{-\Omega(1)}$:\\
(i) If $\|\bx-\bx^*\| \le 2$, then ${\bW\bx}/{|\bW\bx|} = \bx^*$; \\
(ii) If $\|\bx+\bx^*\| \le 2$, then  ${\bW\bx}/{|\bW\bx|} = -\bx^*$.
\end{proposition}

The results in Section \ref{subsec:PM} and \ref{subsec:GPM} characterize the behavior of the iterates generated by Algorithm \ref{alg-pmgpm}. In summary, the sequence $\{\bm{y}^{t}\}_{t\geq 0}$ in the first stage converges to $\bm{u}_1$ up to a sign (Proposition \ref{prop-prop1}). Since $\bm{u}_1$ is close to $\bm{x}^*$ up to a sign (Lemma \ref{lem:eig-W}), the scaled iterate $\sqrt{n}\bm{y}^{t}$ will fall into the region of contraction around $\bm{x}^*$ or $-\bm{x}^*$ after sufficiently many PIs. Then, the PIs are switched to GPIs that will further approach $\bm{x}^*$  or $-\bm{x}^*$ (Proposition \ref{prop:contraction}). The GPI terminates at $\bx^*$ or $-\bm{x}^*$ once the iterate $\bm{x}^t$ goes into the one-step convergence region (Proposition \ref{prop-one-step-conv}). Consequently, we can analyze the overall iteration and time complexities of Algorithm \ref{alg-pmgpm}. The proofs of Theorem \ref{thm-iter-comp} and Corollary \ref{coro-coro2} are deferred to Section \ref{proof-thm-coro} in the supplementary material. 

\section{NUMERICAL EXPERIMENTS}
In this section, we present the numerical results of our method. We validate the efficacy and efficiency of our Algorithm \ref{alg-pmgpm}, which is referred to as SGPI in this section, by comparing it with three other state-of-the-art algorithms for signed graph clustering, which are Signed Ratio Cut (SRC) \citep{kunegis2010spectral}, Signed Positive Over Negative Generalized Eigenproblem (SPONGE) \citep{cucuringu2019sponge}, and Signed Power Means (SPM) \citep{mercado2019spectral}. Instead of using the ture connectivity parameters, we implement SGPI based on the estimated parameters given by \eqref{eq:alpha^+_hat}--\eqref{eq:beta^-_hat} and \eqref{eq:xi-hat}. SPM contains the matrix power $p$ as a hyperparameter. We take two typical values $p=0$ and $p=-10$ for SPM, which are referred to as $\text{SPM}_{p=0}$ and $\text{SPM}_{p=-10}$, respectively. SPONGE has two hyperparameters $\tau^+$ and $\tau^-$, which we set as $\tau^+=10$ and $\tau^-=1$. Our codes are implemented in MATLAB. For SPM, we directly use the MATLAB code provided by \cite{mercado2019spectral}. All the experiments are conducted on a MacBook with 2.3GHz Intel Core i5 CPU and 8GB memory.

\begin{figure*}
\begin{subfigure}[b]{0.190\textwidth}
	\centering
	\includegraphics[width=\textwidth]{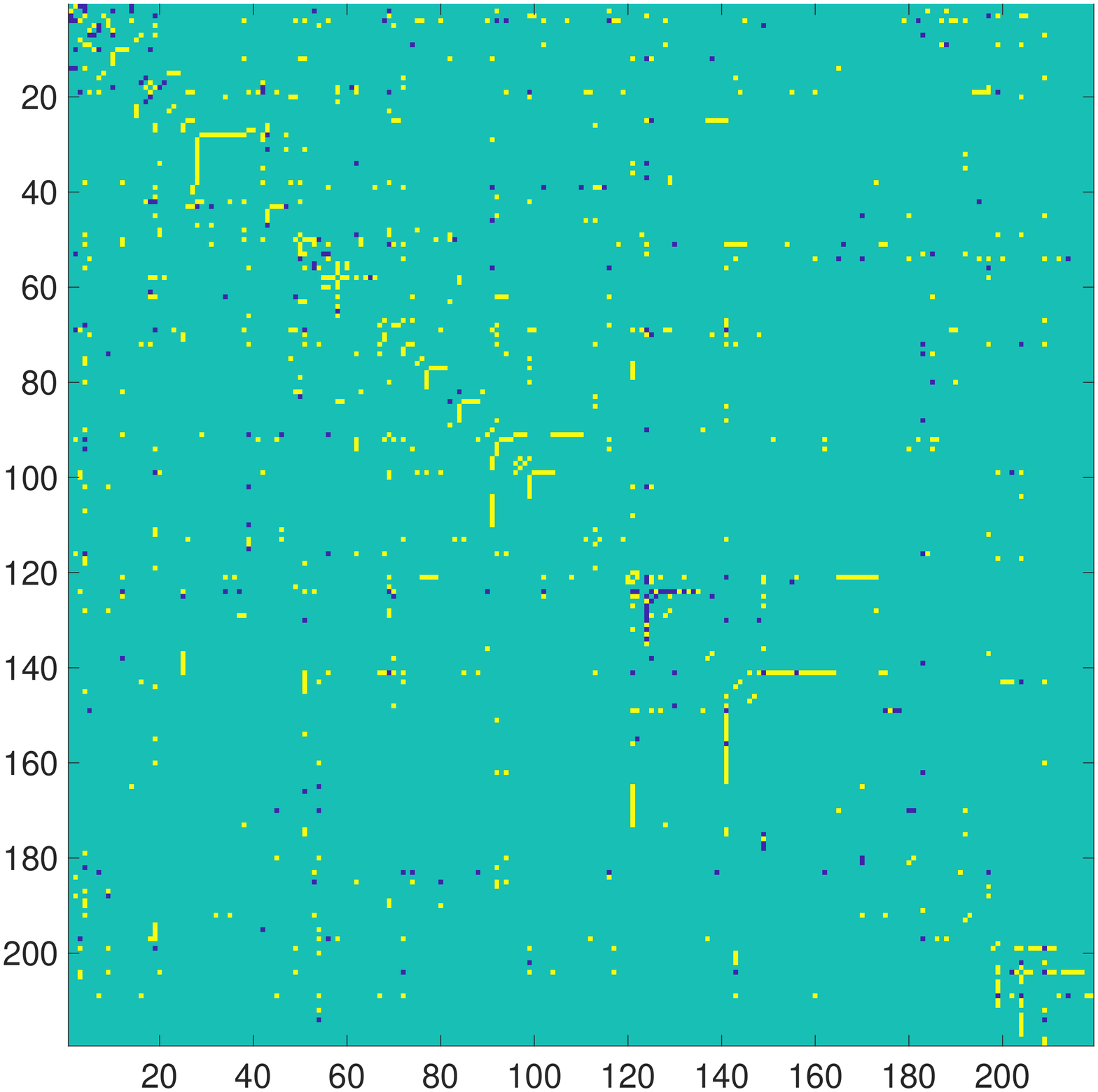}
	\caption{Congress data}
	\label{fig-real-graphs-original}
\end{subfigure}
\hfill
\begin{subfigure}[b]{0.190\textwidth}
	\centering
	\includegraphics[width=\textwidth]{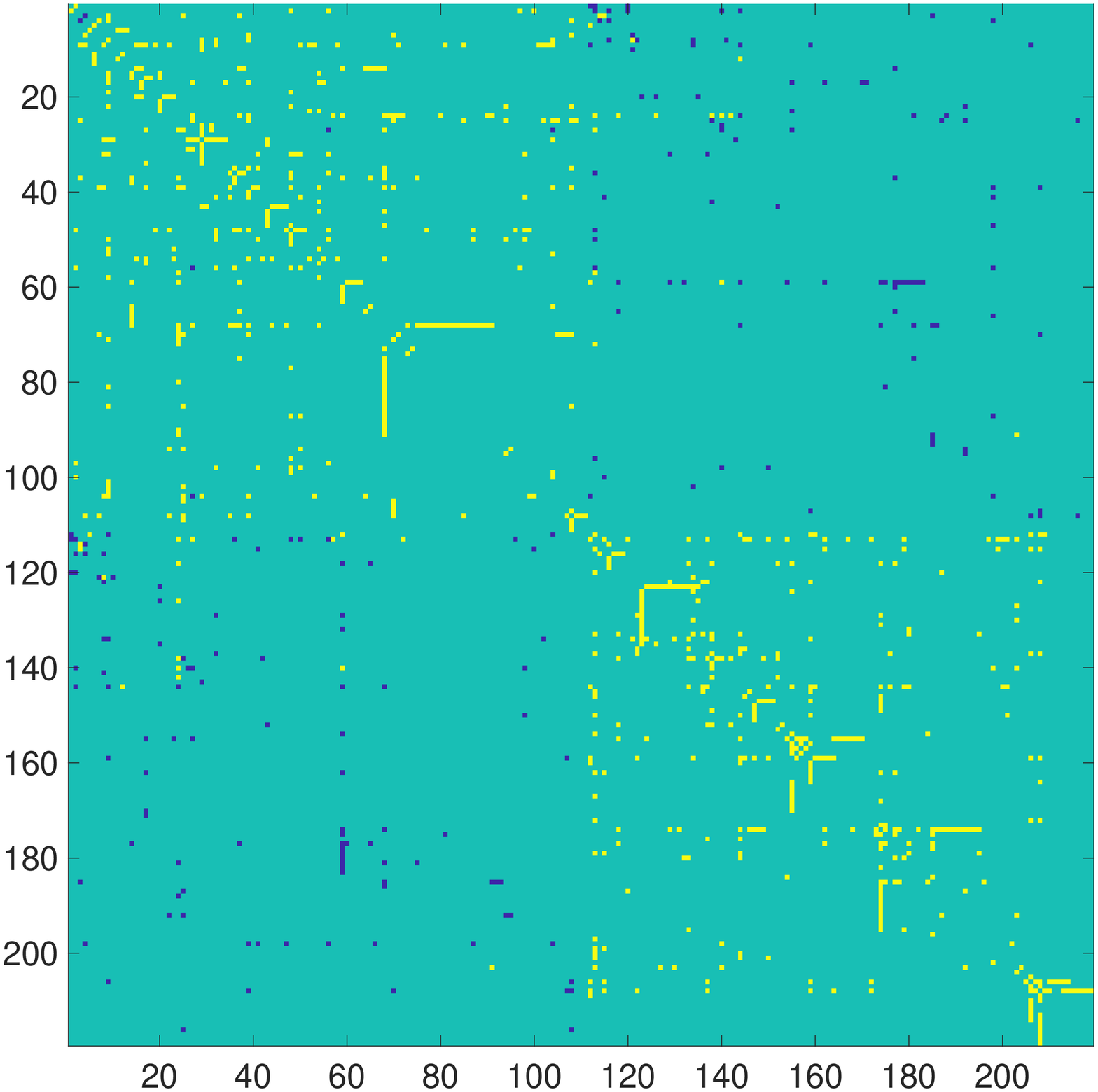}
	\caption{SGPI}
	\label{fig-real-graphs-SGPM}
\end{subfigure}
\hfill
\begin{subfigure}[b]{0.190\textwidth}
	\centering
	\includegraphics[width=\textwidth]{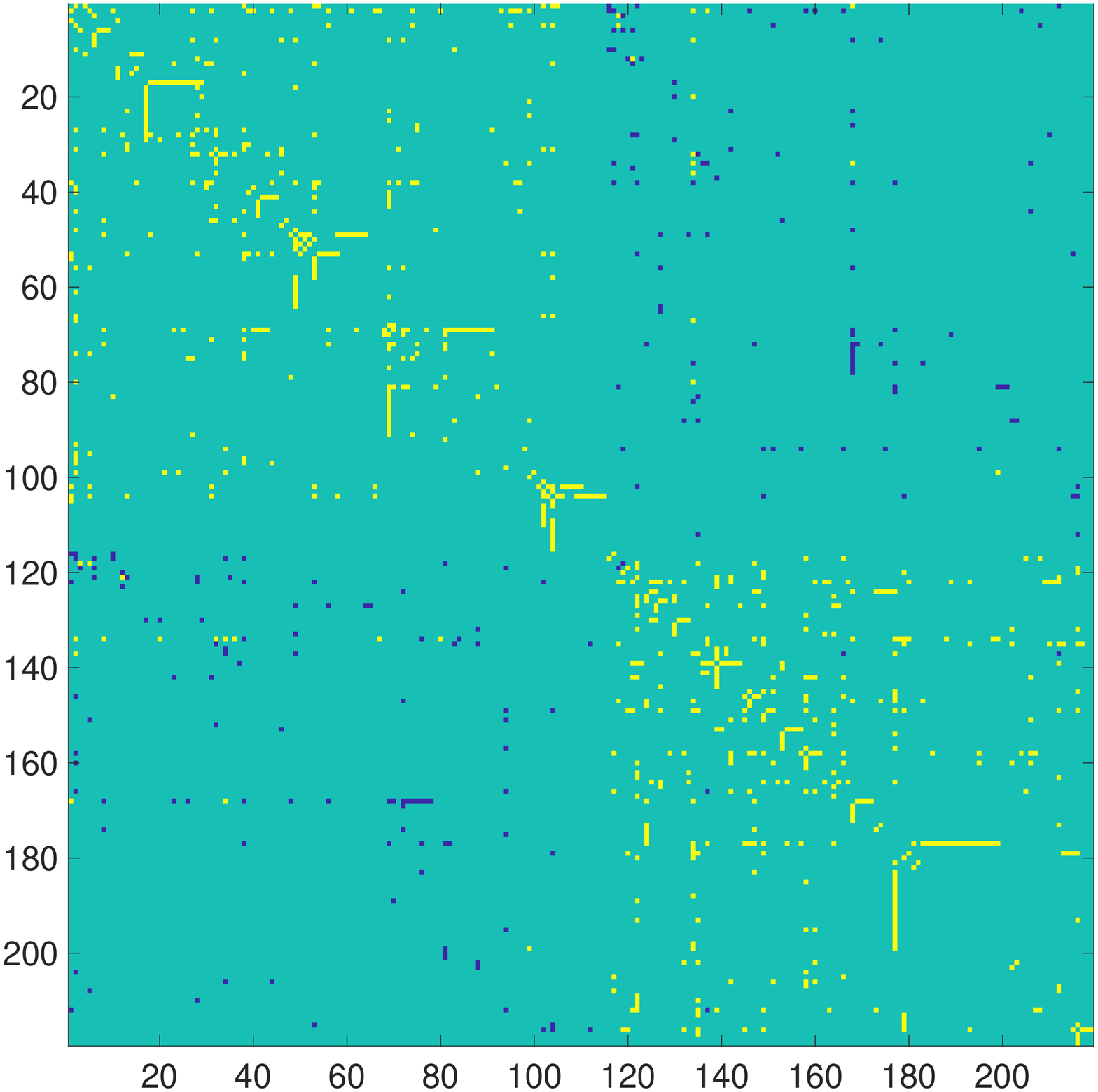}
	\caption{SRC}
	\label{fig-real-graphs-SRC}
\end{subfigure}
\hfill
\begin{subfigure}[b]{0.190\textwidth}
	\centering
	\includegraphics[width=\textwidth]{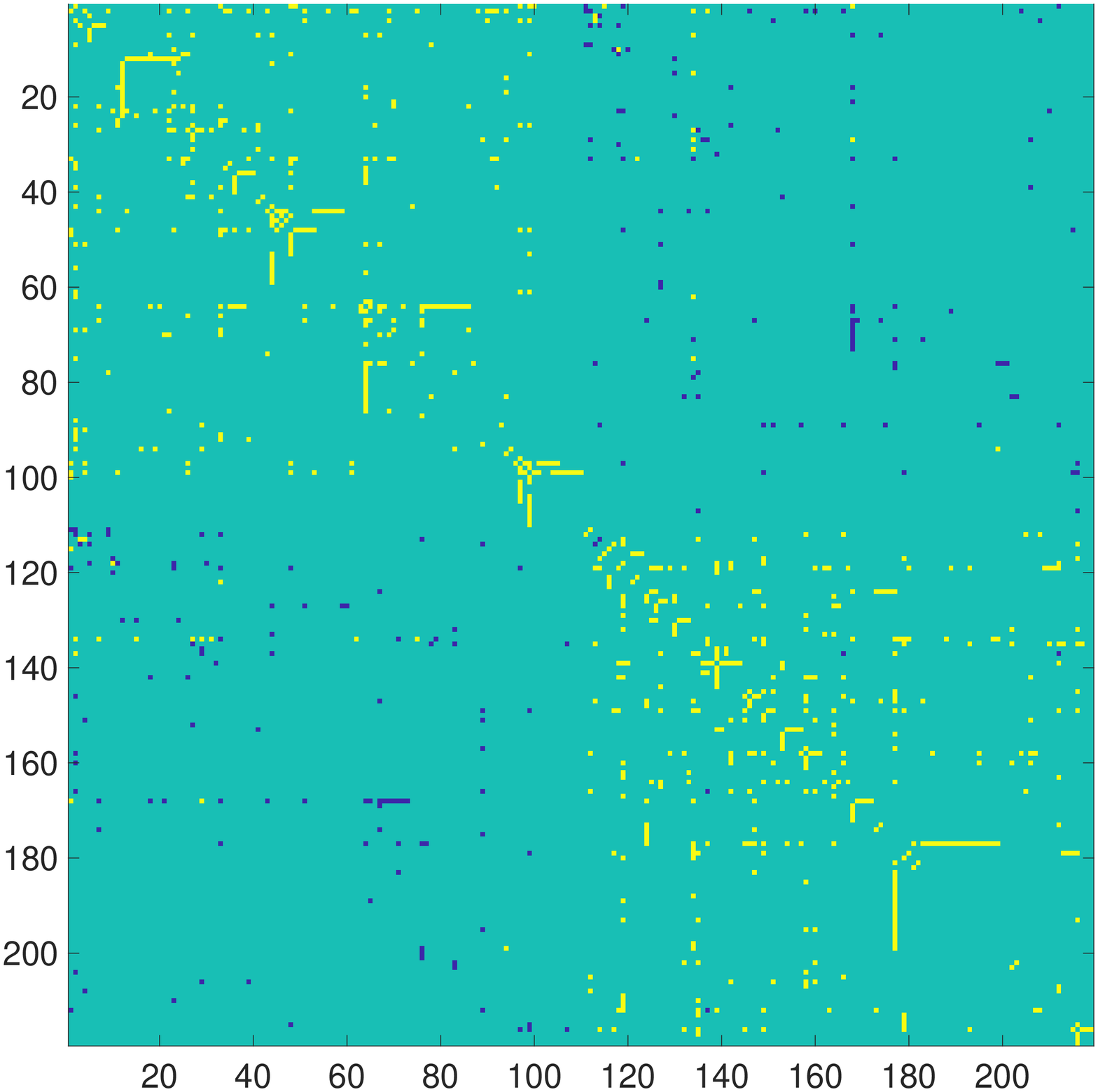}
	\caption{SPONGE}
	\label{fig-real-graphs-SPONGE}
\end{subfigure}
\hfill
\begin{subfigure}[b]{0.190\textwidth}
	\centering
	\includegraphics[width=\textwidth]{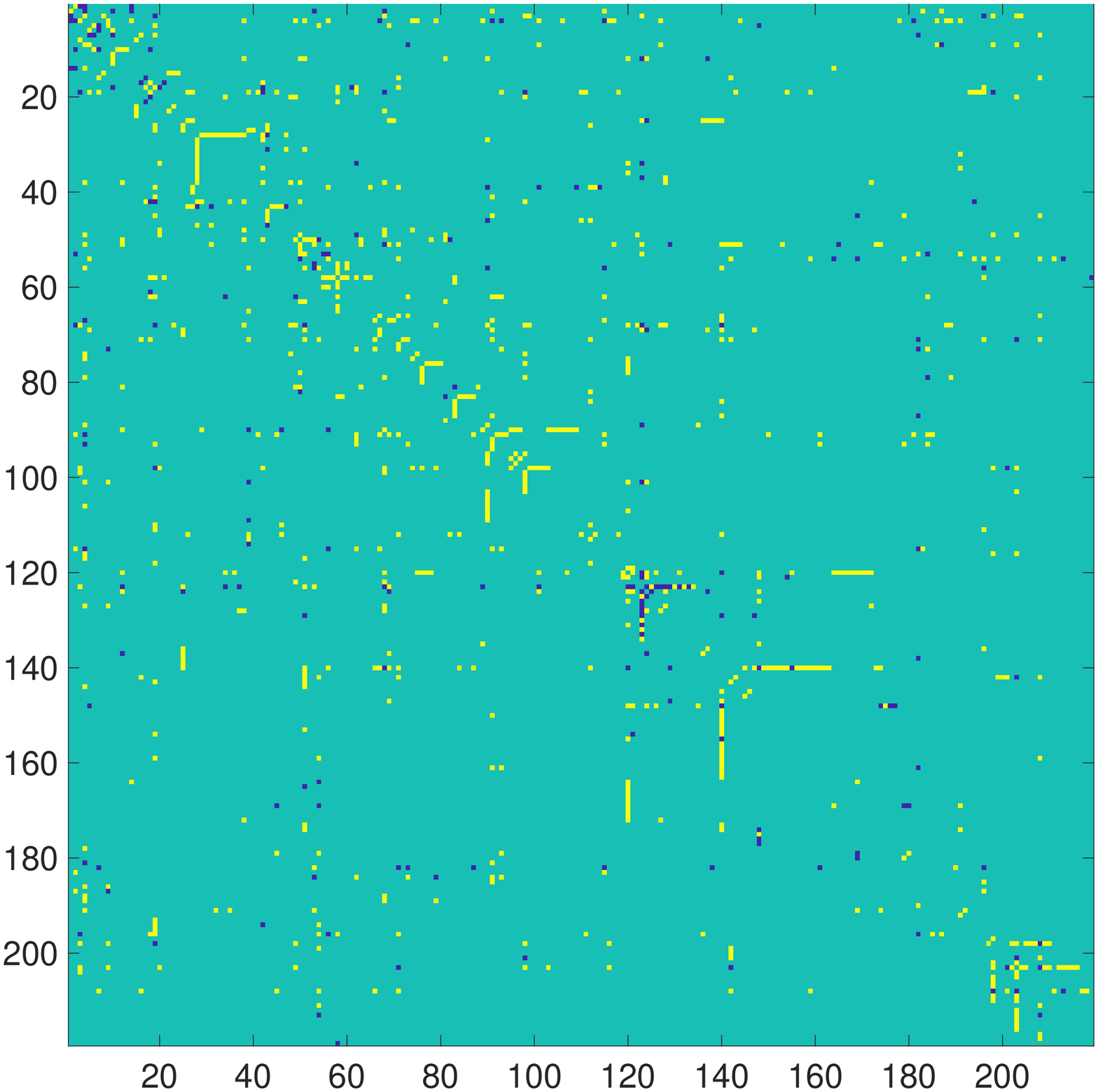}
	\caption{$\text{SPM}_{p=-10}$}
	\label{fig-real-graphs-SPM}
\end{subfigure}
\hfill
\caption{Original adjacency and the sorted adjacency matrices according to the communities identified by different algorithms (the yellow dots represent positive edges, the blue dots represent negative edges, and the green dots mean absence of edges)}
\label{fig-real-graphs}
\end{figure*}

\subsection{Phase Transition}
We first report the community recovery performance under $\textnormal{\textsf{SSBM}}\left(n,\bm{x}^*,\frac{ \alpha^+\log n}{n},\frac{ \alpha^+\log n}{n},\frac{ \alpha^+\log n}{n},\frac{ \alpha^+\log n}{n}\right)$ with $n = 300$. Given $\alpha^+,\alpha^-,\beta^+,\beta^-$, we independently generate 40 signed graphs for each algorithm and compute the ratio of exact recovery. 

Figures \ref{fig-alpha-_vs_beta-} and \ref{fig-alpha+_vs_alpha-} illustrate the heatmaps of exact recovery ratio of all tested algorithms. In Figure \ref{fig-alpha-_vs_beta-}, all signed graphs are generated with fixed $\alpha^+ = 16$ and $\beta^+=9$. The information-theoretic limit \eqref{IT:bound} reduces to $(\sqrt{\alpha^-}-\sqrt{\beta^-})^2\geq1$. It can be observed that most algorithms exhibit obvious phase transition. Particularly, only our SGPI achieves exact recovery down to the information-theoretic limit, which supports the claim in Theorem \ref{thm-iter-comp}. SRC and SPM performs well when $\sqrt{\beta^-}\leq\sqrt{\alpha ^-}-1$, while they can hardly recover the ground-truth if $\sqrt{\beta^-}\geq\sqrt{\alpha ^-}+1$. This is possibly because they equally treats the positive and negative edges of signed graphs, which cannot achieve exact recovery when there are even more negative edges within communities and less positive edges across communities. The superior recovery performance of our algorithm benefits from the weight parameter $\xi$ in our formulation \eqref{RMLE}, which can be accurately estimated from the graphs. 

\begin{table}
	\fontsize{9}{13}\selectfont
	\centering
	\setlength\tabcolsep{5pt}
	\begin{tabular}{l c c c c c c}
		\hline
		$n$ & 100 & 500 & 1000 & 2000 & 5000\\
		\hline
		SGPI & \textbf{0.028} & \textbf{0.39 }& \textbf{2.43} & \textbf{16.32} & \textbf{267.44} \\
		SRC & 0.034 & 0.54 & 3.07 & 23.40 & 312.95 \\
		SPONGE & 0.151 & 1.70 & 8.87 & 66.68 & 832.80 \\
		$\text{SPM}_{p=0}$  & 0.562 & 3.69 & 17.05 & 113.34 & 1720.00 \\
		$\text{SPM}_{p=-10}$ & 0.386 & 2.07 & 9.91 &  65.11 & 964.41 \\
		\hline
	\end{tabular}
	\caption{Total computation time (in seconds) of different community detection algorithms on synthetic data}
	\label{tab-computation time}
\end{table}

In Figure \ref{fig-alpha+_vs_alpha-}, all signed graphs are generated with fixed $\beta^+ = 9$ and $\beta^-=16$. The information-theoretic limit \eqref{IT:bound} reduces to $(\sqrt{\alpha^+}-3)^2+(\sqrt{\alpha^-}-4)^2\geq2$.  We can observe analogous phase transition phenomenon in Figure \ref{fig-alpha+_vs_alpha-}. For the existing algorithms, SRC, SPONGE, and $\text{SPM}_{p=0}$ perform similarly, and $\text{SPM}_{p=-10}$ yields the best recovery performance. In contrast, our SGPI is the only algorithm that achieves exact recovery all the way down to the ring-shaped information-theoretic threshold. 

\subsection{Computational Efficiency}
We compare the computation time of the aforementioned algorithms for community recovery over signed graphs with $\alpha^+=16$, $\alpha^-=9$, $\beta^+=9$, $\beta^-=16$. Given $n$, we independently generate 40 signed graphs for all algorithms and record their overall computation time for completing the community recovery. Specifically, the computation time of SGPI consists of the time consumed by parameter estimation, PIs, and GPIs. In Table \ref{tab-computation time}, we report the total computation time consumed by each algorithm for different $n$. We can observe that SGPI is always the fastest algorithm to achieve community recovery over signed graphs of different sizes. In particular, SGPI exhibits comparable computation time performance with SRC and is substantially faster than other algorithms. Moreover, SGPI achieves the best scalability as $n$ grows. 

\begin{table}
\fontsize{9}{13}\selectfont
\centering
\setlength\tabcolsep{5pt}
\begin{tabular}{l c c c c}
	\hline
	& SGPI & SRC & SPONGE & $\text{SPM}_{p=-10}$ \\
	\hline
	{Congress} & 969 & 968 & 968 & 256 \\
	{Highlandtribes} & 50.5 & 46.9 & 43.9 & 48.0 \\
	\hline
\end{tabular}
\caption{Objective function values of algorithm outputs on real data}
\label{tab-fval}
\end{table}

\subsection{Experiments on Real Data}
We conduct experiments to test the aforementioned algorithms on the Congress dataset \citep{thomas2006get}, which is a publicly available real-world signed network in the KONECT project \citep{kunegis2013konect}.\footnote{konect.cc/networks/convote/} There are 219 nodes in the Congress dataset, whose adjacency matrix is presented in Figure \ref{fig-real-graphs-original}. Each node represents a politician and each edge represents a favorable (labeled by $+1$) or unfavorable (labeled by $-1$) mention of politicians speaking in the United States Congress. We apply the aforementioned algorithms to this signed network to identify the underlying community structure. The sorted adjacency matrices according to the communities identified by  different algorithms are presented in Figures \ref{fig-real-graphs-SGPM} to \ref{fig-real-graphs-SPM}. It can be observed that SGPI detects two obvious communities of similar sizes, where most positive edges appear within communities and most negative edges appear across communities. The sorted adjacency matrices of SRC and SPONGE exhibit similar community structures, while the sorted adjacency matrix of $\text{SPM}_{p=-10}$ exhibits less recognizable community structure.

We also evaluate the objective function of Problem \eqref{RMLE} to assess the performance of the tested algorithms on real datasets. In addition to the Congress dataset, we also make evaluations on the Highlandtribes dataset \citep{read1954cultures} in the KONECT project. The Highlandtribes network contains 16 tribes connected by friendship (corresponding to $+1$ edge) and enmity (corresponding to $-1$ edge). The results are reported in Table \ref{tab-fval}, where a higher function value usually indicates better performance. For the Congress dataset, the quantitative results coincide with the observations from the sorted adjacency matrices illustrated in Figure \ref{fig-real-graphs}.

\section{Conclusion}
In this work, we proposed a new approach for community recovery over signed graphs based on the MLE formulation of the SSBM. We tackled the resulting non-convex and discrete optimization problem via a simple yet efficient two-stage iterative algorithm. Theoretical analysis shows that our method achieves exact recovery within nearly linear time at the information-theoretic limit. Numerical experiments validate the efficacy and efficiency of our approach and show its potential to identify communities in real networks. We leave it as a future work to consider community recovery in the signed SBMs with more than two communities. 

\subsubsection*{Acknowledgements}
This work is supported by the Hong Kong Research Grants Council (RGC) General Research Fund (GRF) Project CUHK 14205421.

\renewcommand\refname{References}
\bibliographystyle{abbrvnat}
\bibliography{aistats_ref}

\clearpage
\appendix

\thispagestyle{empty}

\onecolumn \makesupplementtitle
In the supplementary material, we the provide proofs of the technical results presented in Sections \ref{sec:main} and \ref{sec:proofs}. We first introduce the notation that will be used. 

Given an observed graph $\mG$ generated according to $\textsf{SSBM}(n,\bm{x}^*,p^+,p^-,q^+,q^-)$, we use $N^+$ (resp. $N^-$) to denote the number of positive (resp. negative) edges  and $T^+$ (resp. $T^-$) to denote the number of positive (negative) triangles in $\mG$. In addition, we use $N_{\text{in}}^+$ (resp. $N_{\text{in}}^-$) to denote the number of positive edges (resp. negative edges) within communities and $N_{\text{out}}^+$ (resp. $N_{\text{out}}^-$) to denote positive edges (resp. negative edges) across communities in $\mG$.

\section{Proof of Proposition \ref{thm:MLE}}\label{proof-prop1}
\begin{proof}
	Firstly, we note that
	\begin{align}\label{eq:N+-}
		N^+ = N_{\text{in}}^+ + N_{\text{out}}^+,\quad N^- = N_{\text{in}}^- + N_{\text{out}}^-.
	\end{align}
	Let $\bm{x}\in\{1,-1\}^n$ be a vector whose $i$-th element indicates the community membership of node $i$. For equal-sized community detection, an admissible membership vector should also satisfy $\bm{1}^\top\bm{x}=0$. 
	Then, the likelihood function of $\bm{x}$, i.e., the probability of generating the graph $\mathcal{G}$ based on the community assignment $\bm{x}$, is
	\begin{align*}
		\mathbb{P}(\mG \vert\bx) &\propto (p^+)^{N_{\text{in}}^+} (p^-)^{N_{\text{in}}^-} (1-p^+-p^-)^{\frac{n^2}{4}-\frac{n}{2}-N_{\text{in}}^+-N_{\text{in}}^-}(q^+)^{N_{\text{out}}^+} (q^-)^{N_{\text{out}}^-}
		(1-q^+-q^-)^{\frac{n^2}{4}-N_{\text{out}}^+-N_{\text{out}}^-}. \\ 
		& \propto \left(\frac{p^+}{q^+}\frac{1-q^+-q^-}{1-p^+-p^-}\right)^{N_{\text{in}}^+} \left(\frac{q^-}{p^-}\frac{1-p^+-p^-}{1-q^+-q^-}\right)^{N_{\text{out}}^-}.
	\end{align*}
	Let
	\begin{align*}
		\mu_n  = \log\left(\frac{p^+}{q^+}\frac{1-q^+-q^-}{1-p^+-p^-}\right),\ \nu_n  = \log\left(\frac{q^-}{p^-}\frac{1-p^+-p^-}{1-q^+-q^-}\right),
	\end{align*}
	then, it follows from \eqref{p+-} and \eqref{q+-} that
	\begin{align*}
		\mu_n = \log\left(\frac{\alpha^+}{\beta^+}\right) + \log\left(\frac{n- (\beta^++\beta^-)\log n}{n- (\alpha^++\alpha^-)\log n}\right),\ \nu_n = \log\left(\frac{\beta^-}{\alpha^-}\right) + \log\left(\frac{n- (\alpha^++\alpha^-)\log n}{n- (\beta^++\beta^-)\log n}\right).
	\end{align*}
	Hence, we have
	\begin{align}\label{eq:log-P}
		\log\mathbb{P}(\mathcal{G}\vert\bm{x}) \propto \left(\mu_n N_{\text{in}}^+ + \nu_n N_{\text{out}}^-\right).
	\end{align}
	According to \eqref{eq:N+-}, we note that
	\begin{align*} 
		\bm{x}^\top\bm{A}^+\bm{x} = \sum_{i,j\in[n]}A^+_{ij}x_ix_j = \left(\sum_{i,j\in[n]:x_ix_j=1}A^+_{ij} \right)-
		\left(
		\sum_{i,j\in[n]: x_ix_j=-1}A^+_{ij} \right) = 2N_{\text{in}}^+ - 2N_{\text{out}}^+ = 4N_{\text{in}}^+ - 2N^+ 
	\end{align*}
	and 
	\begin{align*} 
		\bm{x}^\top\bm{A}^-\bm{x} = \sum_{i,j\in[n]}A^-_{ij}x_ix_j = 
		\left(\sum_{i,j\in[n]:x_ix_j=1}A^-_{ij}\right) -
		\left(
		\sum_{i,j\in[n]:x_ix_j=-1}A^-_{ij} \right) 
		= 2N_{\text{in}}^- - 2N_{\text{out}}^- = 2N^- - 4N_{\text{out}}^-.
	\end{align*}
	This, together with \eqref{eq:log-P} and the fact that $N^+,N^-$ are observed, yields the MLE formulation
	\begin{equation*} 
		\max\left\{ \bm{x}^\top(\bm{A}^+-\xi_n\bm{A}^-)\bm{x}:\ \bx \in \{1,-1\}^n,\ \bo^\top\bx = 0 \right\},
	\end{equation*}
	where $\xi_n = \nu_n/\mu_n$.
\end{proof}

\section{Proofs in Section \ref{subsec:esti-para}}

\subsection{Proof of Lemma \ref{lem:edges-triangles}} \label{proof-lemma1}

\begin{proof}
The numbers of edges in $\mG^+$ and $\mG^-$ satisfy
\begin{align}\label{N+-}
	N^+ = \sum_{1 \le i < j \le n}A_{ij}^+,\quad  N^- = \sum_{1 \le i < j \le n}A_{ij}^-,
\end{align}
respectively, where all the $A_{ij}^+$ in $S_1^+$ are i.i.d.~$\mathbf{Bern}(p^+)$, all the $A_{ij}^+$ in $S_2^+$ are i.i.d.~$\mathbf{Bern}(q^+)$ and independent of those in $S_1^+$, all the $A_{ij}^-$ in $S_1^-$ are i.i.d.~$\mathbf{Bern}(p^-)$, and all the $A_{ij}^-$ in $S_2^-$ are i.i.d.~$\mathbf{Bern}(q^-)$ and independent of those in $S_1^-$. Then, according to \eqref{p+-}, \eqref{q+-}, and \eqref{eq:abcd}, we have
\begin{align}
	& \E[N^+] = \left(\frac{n^2}{4}-\frac{n}{2}\right) p^+ + \frac{n^2}{4}q^+ = a(\alpha^+ + \beta^+)-\frac{\alpha^+}{2}\log n, \label{EN+}\\
	& \mathrm{Var}(N^+) = \left(\frac{n^2}{4}-\frac{n}{2}\right) p^+(1-p^+) + \frac{n^2}{4}q^+(1-q^+) \le \frac{n^2}{2}p^+ = \frac{\alpha^+}{2}n\log n. 
\end{align}
Then, by applying Bernstein's inequality for bounded distributions (see, e.g., \citet[Theorem 2.8.4]{vershynin2018high}) to \eqref{N+-}, we obtain
\begin{align*}
	\P\left( \left| N^+  - \E[N^+] \right| \geq \sqrt{n}\log n  \right) & \le 2\exp\left( -\frac{n\log^2n/2}{\alpha^+n\log n/2+\sqrt{n}\log n/3} \right) \\
	& \le 2\exp\left(- \frac{\log n}{  \alpha^+  + 1/\sqrt{n} } \right)\\
	& = 2n^{-\frac{1}{\alpha^+ + 1/\sqrt{n}}}.
\end{align*}
This implies that it holds with probability at least $1- 2n^{-\frac{1}{\alpha^+  + 1/\sqrt{n}}}$ that 
\begin{align*} 
	\left| N^+  - a\left(\alpha^+ + \beta^+\right)+\frac{\alpha^+}{2}\log n \right| \le \sqrt{n}\log n.
\end{align*}
This, together with the triangle inequality, implies
\begin{align*}
	\left| N^+  - a\left(\alpha^+ + \beta^+\right)\right| \le \sqrt{n}\log n+\frac{\alpha^+}{2}\log n.
\end{align*}
Following similar arguments, we can show that it holds with probability at least $1- 2n^{-\frac{1}{\beta^-  + 1/\sqrt{n}}}$ that 
\begin{align*} 
	\left| N^-  - \left(a\alpha^- + b\beta^-\right) \right| \le \sqrt{n}\log n + \frac{\alpha^-}{2}\log n.
\end{align*}

Moreover, the numbers of triangles in $\mG^+$ and $\mG^-$, satisfy
\begin{align}\label{T+}
	T^+ = \frac{1}{6}\tr\left((\bA^+)^3\right),\ T^-_N = \frac{1}{6}\tr\left((\bA^-)^3\right).
\end{align}
respectively. Then, according to \eqref{p+-}, \eqref{q+-}, and \eqref{eq:abcd}, one can verify that
\begin{align}
	\E[T^+] & = \frac{n}{6}\left(\frac{n}{2}-1\right)\left(\frac{n}{2}-2\right)\left(p^{+}\right)^3 + \frac{n^2}{4}\left(\frac{n}{2}-1\right)p^{+}\left(q^+\right)^2 \notag\\
	&= \frac{\log^3n}{24}\left(1-\frac{6}{n}+\frac{8}{n^2}\right)(\alpha^{+})^3 + \frac{\log^3n}{8}\left( 1-\frac{2}{n}\right)\alpha^+(\beta^+)^2 \notag\\
	&= b\left((\alpha^+)^3 + 3\alpha^+(\beta^+)^2\right) - \frac{\log^3 n}{2n}\left((\alpha^+)^3 + \alpha^+(\beta^+)^2\right) + \frac{\log^3n}{3n^2}(\alpha^+)^3 \label{ET+0}\\
	&\le \frac{(\alpha^+)^3}{6}\log^3n.\label{ET+}
\end{align}
where the inequality follows from $\alpha^+>\beta^+$. This, together with \citet[Theorem 4.2]{vu2002concentration} with $k=3$, $\tilde{k}=2$, $\mathscr{E}_0=\E[T^+]$, $\mathscr{E}_1=\sqrt{\E[T^+]}$, $\mathscr{E}_2=1$, $\lambda=\log n$, $c_3 = 8\sqrt{6}+2\sqrt{3} < 24$, and $d_3=14$, yields 
\begin{align*}
	\P\left( \left| T^+  - \E[T^+]   \right| \ge 24\sqrt{\log n} \left(\E[T^+]\right)^{3/4}  \right) \le 14\exp\left( -\frac{\log n}{4} \right) = 14n^{-\frac{1}{4}}.
\end{align*}
That is, it holds with probability at least $1-14n^{-\frac{1}{4}}$ that
\begin{align*}
	\left| T^+  -\E[T^+]  \right| \le 24\sqrt{\log n} \left(\E[T^+]\right)^{3/4} \le 7(\alpha^+)^{\frac{9}{4}}\left(\log n\right)^{\frac{11}{4}},
\end{align*}
where the last inequality is due to \eqref{ET+}. This, together with \eqref{ET+0} and the triangle inequality, gives
\begin{align*}
	\left| T^+  - b\left((\alpha^{+})^3  + 3\alpha^+(\beta^+)^2  \right)  \right| &\leq \left| T^+  -\E[T^+]  \right| + \frac{\log^3 n}{2n}\left((\alpha^+)^3 + \alpha^+(\beta^+)^2\right) - \frac{\log^3n}{3n^2}(\alpha^+)^3 \\ 
	&\le  7(\alpha^+)^{\frac{9}{4}}\left(\log n\right)^{\frac{11}{4}}  + \frac{ \alpha^{+^3}\log^3n}{2n} \le 8(\alpha^+)^{\frac{9}{4}}\left(\log n\right)^{\frac{11}{4}},
\end{align*}
where the last inequality holds for $n \ge \alpha^{+}$. Following similar arguments, when $n \ge \beta^-$, it holds with probability at least $1-14n^{-\frac{1}{4}}$ that 
\begin{align*}
	\left| T^-  - b\left( (\alpha^{-})^3 + 3\alpha^-(\beta^-)^2 \right)  \right| \le  8(\beta^-)^{\frac{9}{4}}\left(\log n\right)^{\frac{11}{4}}.
\end{align*}
Finally, by applying the union bound, the desired results holds with probability at least $1- 2n^{-\frac{1}{\alpha^+  + 1/\sqrt{n}}}-2n^{-\frac{1}{\beta^-  + 1/\sqrt{n}}}-28n^{-\frac{1}{4}}$. 
\end{proof}

\subsection{Proof of Lemma \ref{lem:unique-sol}} \label{proof-lemma2}
\begin{proof}
Indeed, motivated by \eqref{EN+} and \eqref{ET+0}, we let the estimators of the connectivity parameters $\alpha^+$ and $\beta^+$ be the solutions to the following system of equations:
\begin{align}\label{eq:positive}
	\begin{cases}
		& a\hat{\alpha}^+ + a\hat{\beta}^+  = N^+,   \\
		& b(\hat{\alpha}^{+})^3 + 3b\hat{\alpha}^+(\hat{\beta}^+)^2 = T^+.
	\end{cases}
\end{align} 
Similarly, we let the estimators of the connectivity parameters $\alpha^-$ and $\beta^-$ be the solutions to the following system of equations:
\begin{align}\label{eq:negative}
	\begin{cases}
		& a \hat{\alpha}^- + a\hat{\beta}^-  = N^-,   \\
		& b  (\hat{\alpha}^{-})^3 + 3b\hat{\alpha}^-(\hat{\beta}^-)^2 = T^-.
	\end{cases}
\end{align}
We can consider the following unified systems of equations for given $C,D \in \R$:
\begin{align}\label{sys:cubic}
	\begin{cases}
		& a x+ ay = C,\\
		& b x^3 + 3b xy^2 = D,
	\end{cases}
\end{align} 
Eliminating $y$ in the above system yields
\begin{align*}
	4bx^3 - \frac{6bC}{a}x^2 + \frac{3bC^2}{a^2}x - D = 0.
\end{align*}
According to \eqref{eq:abcd}, we further have
\begin{align*}
	\left(\log^3n\right)x^3 - \frac{6C\log^2n}{n}x^2 + \frac{12C^2\log n}{n^2}x-6D = 0.
\end{align*}
This is equivalent to
\begin{align*}
	\left( x\log n - \frac{2C}{n}\right)^3 = 6D - \frac{8C^3}{n^3}.
\end{align*}
This implies that the solution of the system \eqref{sys:cubic} has unique real root:
\begin{align}\label{eq:sol-cubic}
	x = \frac{1}{\log n}\left(\frac{2C}{n} + \sqrt[3]{6D - \frac{8C^3}{n^3}}\right),\ y = \frac{C}{a}-x.
\end{align}
To proceed, let
\begin{align*}
	\bar{N}^+=a(\alpha^++\beta^+),\ \bar{T}^+=b\alpha^{+^3}+3b\alpha^+\beta^{+^2},\ \delta_1 = N^+ - \bar{N}^+,\ \delta_2 = T^+ - \bar{T}^+.
\end{align*}
By letting $C=\bar{N}^+$, $D=\bar{T}^+$ in \eqref{sys:cubic}, it is obvious that the real root is $x=\alpha^+,y=\beta^+$ and
\begin{align}\label{eq:alpha+}
	\alpha^+ = \frac{1}{\log n}\left(\frac{2\bar{N}^+}{n} + \sqrt[3]{6\bar{T}^{+} - \frac{8\bar{N}^{+^3}}{n^3}}\right).
\end{align}
In particular, it follows from \eqref{eq:abcd} that 
\begin{align}\label{eq1:lem:unique-sol}
	\sqrt[3]{6\bar{T}^+ - \frac{8\bar{N}^{+^3}}{n^3}}  = \sqrt[3]{6\left( b\alpha^{+^3}+3b\alpha^+\beta^{+^2} \right)- \frac{8\left(a(\alpha^++\beta^+)\right)^3}{n^3}} = \frac{\alpha^+-\beta^+}{2}\log n \ge 0.
\end{align}
Moreover, by letting $C=N^+$ and $D=T^+$ in \eqref{sys:cubic}, it follows from \eqref{eq:sol-cubic} that
\begin{align}\label{eq3:lem:unique-sol}
	\hat{\alpha}^+ = \frac{1}{\log n}\left(\frac{2N^+}{n} + \sqrt[3]{6T^+ - \frac{8N^{+^3}}{n^3}}\right). 
\end{align}
By Lemma \ref{lem:edges-triangles}, we have $|\delta_1| \lesssim \sqrt{n}\log n $ and $|\delta_2| \lesssim \left(\log n\right)^{\frac{11}{4}}$. This, together with $\bar{N}^+ \asymp n\log n$ and $\bar{T}^+\asymp \log^3n$, implies that for sufficiently large $n$,
\begin{align}\label{eq2:lem:unique-sol}
	\sqrt[3]{6T^+ - \frac{8N^{+^3}}{n^3}} = \sqrt[3]{6(\bar{T}^++\delta_2) - \frac{8(\bar{N}^++\delta_1)^3}{n^3}} \ge 0
\end{align}
Then, we have
\begin{align*}
	\left| \sqrt[3]{6\bar{T}^+ - \frac{8\bar{N}^{+^3}}{n^3}} -  \sqrt[3]{6T^+ - \frac{8N^{+^3}}{n^3}} \right|  &\le \left| \sqrt[3]{6(\bar{T}^+-T^+) + \left(\frac{8N^{+^3}}{n^3} - \frac{8\bar{N}^{+^3}}{n^3} \right) } \right| \\ 
	& \le \sqrt[3]{6\delta_2 +  \frac{8\delta_1\left((\bar{N}^++\delta_1)^2+(\bar{N}^++\delta_1)\bar{N}^++\bar{N}^{+^2}\right)}{n^3}} \\
	& \le \sqrt[3]{6 |\delta_2|} + \frac{2}{n}\sqrt[3]{|\delta_1|(2\bar{N}^++\delta_1)^2} \\
	&\leq 4(\alpha^+)^{3/4}(\log n)^{11/12} + 4\sqrt[3]{\alpha^++\beta^+}\frac{\log n}{\sqrt[6]{n}},
\end{align*}
where the first inequality is due to \eqref{eq1:lem:unique-sol}, \eqref{eq2:lem:unique-sol}, and the fact that $|\sqrt[3]{u} - \sqrt[3]{v}| \le |\sqrt[3]{u-v}|$ for any $u,v\ge 0$ and the last inequality follows from Lemma \ref{lem:edges-triangles}. 
This, together with \eqref{eq:alpha+} and \eqref{eq3:lem:unique-sol}, gives
\begin{align*}
	\left| \alpha^+ - \hat{\alpha}^+\right|  &= \frac{1}{\log n}\left| \frac{2(\bar{N}^+-N^+)}{n} + \sqrt[3]{6\bar{T}^+ - \frac{8\bar{N}^{+^3}}{n^3}} -  \sqrt[3]{6T^+ - \frac{8N^{+^3}}{n^3}} \right| \\
	&\le 3n^{-1/2} + 4(\alpha^+)^{3/4}(\log n)^{-1/12} + 4\sqrt[3]{\alpha^++\beta^+}n^{-1/6}.
\end{align*}
for sufficiently large $n$. Moreover, we have
\begin{align*}
	\left| \beta^+ - \hat{\beta}^+\right| = \left| \frac{\bar{N}^+-N^+}{a} - (\alpha^+-\hat{\alpha}^+) \right| 
	\le 6n^{-1/2} + 4(\alpha^+)^{3/4}(\log n)^{-1/12} + 4\sqrt[3]{\alpha^++\beta^+}n^{-1/6}.
\end{align*}
Thus, there exists $\kappa_1>0$, such that for sufficiently large $n$,
\begin{align*}
	& \max\left\{ | \alpha^+ - \hat{\alpha}^+|, | \beta^+ - \hat{\beta}^+|\right\}  \le \kappa_1\left(\log n \right)^{-\frac{1}{12}}.
\end{align*}
Following similar arguments, we can show that there exists $\kappa_2>0$, such that for sufficiently large $n$, 
\begin{align*}
	& \max\left\{ | \alpha^- - \hat{\alpha}^-|, | \beta^- - \hat{\beta}^-|\right\}  \le \kappa_2\left(\log n \right)^{-\frac{1}{12}}.
\end{align*}
\end{proof}

\section{Proof of Proposition \ref{prop:esti-p-q}}\label{proof-prop2}
\begin{proof}
By Lemma \ref{lem:unique-sol}, for $n>\max\{\alpha^+,\beta^-\}$, it holds with probability at least $1-n^{-\Omega(1)}$ that
\begin{align*}
	| \alpha^+ - \hat{\alpha}^+ | \leq \kappa_1\left(\log n \right)^{-\frac{1}{12}},
\end{align*}
which implies that 
\begin{align*}
	\log\left(1 - \frac{\kappa_1}{\alpha^+}\left(\log n \right)^{-\frac{1}{12}}\right) \leq \log\left(\frac{\hat{\alpha}^+}{\alpha^+}\right) \leq \log\left(1 + \frac{\kappa_1}{\alpha^+}\left(\log n \right)^{-\frac{1}{12}}\right),
\end{align*}
for sufficiently large $n$ such that $1 - \frac{\kappa_1}{\alpha^+}\left(\log n \right)^{-\frac{1}{12}}>0$. Then, for sufficiently large $n$ such that $1 - \frac{\kappa_1}{\alpha^+}\left(\log n \right)^{-\frac{1}{12}}>\frac{1}{2}$, we have
\begin{align}
	\zeta_1^+ 
	&\coloneqq |\log\hat{\alpha}^+-\log\alpha^+| 
	= \left|\log\left(\frac{\hat{\alpha}^+}{\alpha^+}\right)\right| \nonumber\\
	&\leq \max\left\{ \left|\log\left(1 - \frac{\kappa_1}{\alpha^+}\left(\log n \right)^{-\frac{1}{12}}\right)\right|, \left|\log\left(1 + \frac{\kappa_1}{\alpha^+}\left(\log n \right)^{-\frac{1}{12}}\right)\right|\right\} \nonumber\\ 
	&\leq \frac{2\kappa_1}{\alpha^+}\left(\log n \right)^{-\frac{1}{12}}, \label{eq:log-alphahat-alpha+}
\end{align}
where the second inequality is because 
$|\log(1-x)|=\log\left(1+\frac{x}{1-x}\right)\leq\frac{x}{1-x}\leq2x$ for $0<x\leq\frac{1}{2}$ and $|\log(1+x)|\leq x$ for $x>0$. Following similar arguments for \eqref{eq:log-alphahat-alpha+}, we can obtain
\begin{align}
	\zeta_2^+ 
	&\coloneqq |\log\hat{\beta}^+-\log\beta^+| 
	\leq \frac{2\kappa_2}{\beta^+}\left(\log n \right)^{-\frac{1}{12}}, \label{eq:log-betahat-beta+}\\
	\zeta_1^- 
	&\coloneqq |\log\hat{\alpha}^--\log\alpha^-| 
	\leq \frac{2\kappa_1}{\alpha^-}\left(\log n \right)^{-\frac{1}{12}},\label{eq:log-alphahat-alpha}\\
	\zeta_2^- 
	&\coloneqq |\log\hat{\beta}^--\log\beta^-| 
	\leq \frac{2\kappa_2}{\beta^-}\left(\log n \right)^{-\frac{1}{12}}.\label{eq:log-betahat-beta-}
\end{align}
Therefore,
\begin{align*}
	|\hat{\xi}-\xi| 
	&= \left|\frac{\log(\hat{\beta}^-/\hat{\alpha}^-)}{\log(\hat{\alpha}^+/\hat{\beta}^+)} - \frac{\log({\beta}^-/{\alpha}^-)}{\log({\alpha}^+/{\beta}^+)}\right|\\
	&=\left|\frac{\log\left(\frac{\alpha^+}{\beta^+}\right)\left(\log\left(\frac{\hat{\beta}^-}{\hat{\alpha}^-}\right)-\log\left(\frac{\beta^-}{\alpha^-}\right)\right) - \log\left(\frac{\beta^-}{\alpha^-}\right)\left(\log\left(\frac{\hat{\alpha}^+}{\hat{\beta}^+}\right)-\log\left(\frac{\alpha^+}{\beta^+}\right)\right)}{\log\left(\frac{\alpha^+}{\beta^+}\right)\log\left(\frac{\hat{\alpha}^+}{\hat{\beta}^+}\right)}\right| \\
	&=\left|\frac{\log\left(\frac{\alpha^+}{\beta^+}\right)\left(\log\hat{\beta}^--\log\beta^-+\log\alpha^--\log\hat{\alpha}^-\right) - \log\left(\frac{\beta^-}{\alpha^-}\right)\left(\log\hat{\alpha}^+-\log\alpha^++\log\beta^+-\log\hat{\beta}^+\right)}{\log\left(\frac{\alpha^+}{\beta^+}\right)\left(\log\hat{\alpha}^+-\log\hat{\beta}^+\right)}\right| \\
	&\leq \left|\frac{\log\left(\frac{\alpha^+}{\beta^+}\right)}{\log\left(\frac{\alpha^+}{\beta^+}\right)\left(\log\hat{\alpha}^+-\log\hat{\beta}^+\right)}\right| \left( \zeta_2^- + \zeta_1^- \right) 
	+ \left|\frac{\log\left(\frac{\beta^-}{\alpha^-}\right)}{\log\left(\frac{\alpha^+}{\beta^+}\right)\left(\log\hat{\alpha}^+-\log\hat{\beta}^+\right)}\right| \left( \zeta_1^+ + \zeta_2^+\right) \\ 
	&\leq \frac{1}{\log(\alpha^+/\beta^+)} \left(\frac{2\kappa_2}{\beta^-}\left(\log n \right)^{-\frac{1}{12}} + \frac{2\kappa_1}{\alpha^-}\left(\log n \right)^{-\frac{1}{12}}\right) + \frac{\xi}{\log(\alpha^+/\beta^+)} \left(\frac{2\kappa_1}{\alpha^+}\left(\log n \right)^{-\frac{1}{12}} + \frac{2\kappa_2}{\beta^+}\left(\log n \right)^{-\frac{1}{12}} \right) \\
	&= \kappa\left(\log n \right)^{-\frac{1}{12}} 
\end{align*}
where the first inequality follows from triangular inequality, the second inequality holds due to \eqref{eq:log-alphahat-alpha+}--\eqref{eq:log-betahat-beta-} and the fact that for sufficiently large $n$ such that $\log\left(\frac{\alpha^+}{\beta^+}\right) - \left(\frac{2\kappa_1}{\alpha^+}+\frac{2\kappa_2}{\beta^+}\right) (\log n)^{-\frac{1}{12}}>0$, we have
\begin{align*}
	\left|\frac{\log\left(\frac{\alpha^+}{\beta^+}\right)}{\log\left(\frac{\alpha^+}{\beta^+}\right)\left(\log\hat{\alpha}^+-\log\hat{\beta}^+\right)}\right| &= \left|\frac{1}{\log\left(\frac{\alpha^+}{\beta^+}\right) + (\log\hat{\alpha}^+ - \log{\alpha}^+) - (\log\hat{\beta}^+ - \log{\beta}^+)} \right| \\
	&\leq \frac{1}{\log\left(\frac{\alpha^+}{\beta^+}\right) - \left(\frac{2\kappa_1}{\alpha^+}+\frac{2\kappa_2}{\beta^+}\right) (\log n)^{-\frac{1}{12}}} < \frac{1}{\log(\alpha^+/\beta^+)},
\end{align*}
and 
\begin{align*}
	\left|\frac{\log\left(\frac{\beta^-}{\alpha^-}\right)}{\log\left(\frac{\alpha^+}{\beta^+}\right)\left(\log\hat{\alpha}^+-\log\hat{\beta}^+\right)}\right| &= \left|\frac{\xi}{\log\left(\frac{\alpha^+}{\beta^+}\right) + (\log\hat{\alpha}^+ - \log{\alpha}^+) - (\log\hat{\beta}^+ - \log{\beta}^+)}\right| \\
	&\leq \frac{\xi}{\log\left(\frac{\alpha^+}{\beta^+}\right) - \left(\frac{2\kappa_1}{\alpha^+}+\frac{2\kappa_2}{\beta^+}\right) (\log n)^{-\frac{1}{12}}}
	< \frac{\xi}{\log(\alpha^+/\beta^+)},
\end{align*}
and the last equality directly follows by setting
\begin{align*}
	\kappa \coloneqq \frac{1}{\log(\alpha^+/\beta^+)}\left( \frac{2\kappa_2}{\beta^-} + \frac{2\kappa_1}{\alpha^-} \right) + \frac{\xi}{\log(\alpha^+/\beta^+)} \left( \frac{2\kappa_1}{\alpha^+} + \frac{2\kappa_2}{\beta^+} \right).
\end{align*}
The desired result is proved.
\end{proof}

\section{Proofs in Section \ref{subsec:PM}}
\subsection{Required Technical Lemmas}
We first present some technical lemmas that will be used in the sequel. The following lemma guarantees that the constant $c_0$ that appears throughout this paper is nonzero, which is necessary for many related results to remain valid.
\begin{lemma}
If $\alpha^+\neq\beta^+$, then $c_0=(\alpha^+-\beta^+)-\xi(\alpha^--\beta^-) \neq 0$
\end{lemma}
\begin{proof}
Suppose that $c_0=0$, then $\alpha^+\neq\beta^+$ implies $\alpha^-\neq\beta^-$. By \eqref{xi}, we have
\begin{align}\label{eq:fraction}
	\frac{\alpha^+-\beta^+}{\alpha^--\beta^-} = \frac{\log\beta^--\log\alpha^-}{\log\alpha^+-\log\beta^+}.
\end{align}
\begin{itemize}
	\item[$1^\circ$] If $\frac{\alpha^+-\beta^+}{\alpha^--\beta^-} >0$ , then we have
	\begin{align*}
		& \alpha^+-\beta^+ \text{\ and\ } \alpha^--\beta^- \text{\ have the same sign\ }\\
		\Leftrightarrow\ & \log\alpha^+-\log\beta^+ \text{\ and\ } \log\alpha^--\log\beta^- \text{\ have the same sign\ }\\
		\Leftrightarrow\ &\frac{\log\alpha^+-\log\beta^+}{\log\alpha^--\log\beta^-} > 0 
		\Leftrightarrow\  \frac{\log\beta^--\log\alpha^-}{\log\alpha^+-\log\beta^+} < 0,
	\end{align*}
	which contradicts \eqref{eq:fraction}.
	\item[$2^\circ$] If $\frac{\alpha^+-\beta^+}{\alpha^--\beta^-} <0$, similar arguments yield a contradiction.
\end{itemize}
As a consequence, $c_0\neq 0$.
\end{proof}

\begin{lemma}\label{lem:rho}
Suppose that Assumption \ref{assump:base} holds and $\rho= \bo^\top\tbA\bo/n^2$. Then, it holds with probability at least $1-2n^{-\frac{1}{2(\alpha^++\beta^++1)}}-2n^{-\frac{1}{2(\alpha^- + \beta^-+1)}}$ that
\begin{align}\label{rst:lem:rho}
	\left| \rho - \E[\rho] \right|
	\leq \frac{2\log n}{n^{3/2}},
\end{align}
where $\E[\rho] = \frac{\left(p^++q^+\right) - \xi\left(p^-+q^-\right)}{2}-\frac{p^+-\xi p^-}{n}$.
\end{lemma}
\begin{proof}
For ease of exposition, we define
\begin{align*}
	\rho^+ = \frac{\bo^\top\bA^+\bo}{n^2},\quad \rho^- = \frac{\bo^\top\bA^-\bo}{n^2}.
\end{align*}
Since Assumption \ref{assump:base} holds, then by \eqref{distri:a-same} and \eqref{distri:a-diff}, we can verify that 
\begin{align*}
	\mathbb{E}[\rho^+] = \frac{p^++q^+}{2}-\frac{p^+}{n},\quad \mathbb{E}[\rho^-] = \frac{p^-+q^-}{2}-\frac{p^-}{n}.
\end{align*}
Thus, according to \citet[Lemma 1]{wang2020nearly}, we obtain
\begin{align*}
	\P\left( \left| \rho^+ - \mathbb{E}[\rho^+]\right| \le \frac{\log n}{n^{3/2}}\right)
	= \P\left( \left| \rho^+ - \left(\frac{1}{2}(p^++q^+)-\frac{p^+}{n} \right)\right| 
	\le \frac{\log n}{n^{3/2}}\right) \ge 1 - 2n^{-\frac{1}{2(\alpha^++\beta^++1)}}
\end{align*}
and
\begin{align*}
	\P\left( \left| \rho^- - \mathbb{E}[\rho^-]\right| \le \frac{\log n}{n^{3/2}}\right)
	= \P\left( \left| \rho^- - \left(\frac{1}{2}(p^- + q^-)-\frac{p^-}{n} \right)\right| \le \frac{\log n}{n^{3/2}}\right) \ge 1 - 2n^{-\frac{1}{2(\alpha^- + \beta^-+1)}}.
\end{align*}
These, together with $\rho = \rho^+ - \xi \rho^-$ and the union bound, imply that it holds with probability at least $1-2n^{-\frac{1}{2(\alpha^++\beta^++1)}}-2n^{-\frac{1}{2(\alpha^- + \beta^-+1)}}$ that
\begin{align}\label{eq1:lem-rho}
	\left| \rho - \E[\rho] \right| \le \left| \rho^+ - \E[\rho^+] \right| + \xi \left| \rho^- - \E[\rho^-] \right|  \le \frac{(1+\xi)\log n}{n^{3/2}},
\end{align}
where
\begin{align}\label{eq:E-rho}
	\E[\rho] = \frac{1}{2}\left(p^++q^+\right) - \frac{\xi}{2}\left(p^-+q^-\right)-\frac{1}{n}\left(p^+-\xi p^-\right).
\end{align}
This complete the proof. 
\end{proof}

Next, we present a spectral bound on the deviation of $\tbA$ from its mean.  
\begin{lemma}\label{lem:Delta}
Suppose that Assumption \ref{assump:base} holds. Then, it holds with probability at least $1-2n^{-3}$ that
\begin{align}\label{rst:lem:Delta}
	\|\tbA - \E[\tbA]\| \leq c_1\sqrt{\log n},
\end{align}
where $c_1 > 0$ is a constant.
\end{lemma}
\begin{proof}
According to \citet[Theorem 5.2]{lei2015consistency} and the union bound, it holds with probability at least $1-2n^{-3}$ for some constants $c_1^+,c_1^- >0$ that 
\begin{align*}
	\|\bm{A}^+-\mathbb{E}[\bm{A}^+]\| \le c_1^+ \sqrt{\log n},\quad \|\bm{A}^- - \mathbb{E}[\bm{A}^-]\| \le c_1^- \sqrt{\log n}. 
\end{align*}
These, together with $\widetilde{\bm{A}} = \bm{A}^+-\xi \bm{A}^-$, yield
\begin{align*}
	\|\tbA - \E[\tbA]\| &= \|\bm{A}^+-\xi \bm{A}^- - \mathbb{E}[\bm{A}^+-\xi \bm{A}^-]\| \\
	&\leq \|\bm{A}^+-\mathbb{E}[\bm{A}^+]\| + \xi  \|\bm{A}^--\mathbb{E}[\bm{A}^-]\| \\
	&\leq c_1^+ \sqrt{\log n} + \xi  c_1^-\sqrt{\log n} \\
	&= c_1\sqrt{\log n},
\end{align*}
where the last equality holds by letting $c_1 = c_1^+ + \xi c_1^- = c_1^+ + \frac{\log(\beta^-/\alpha^-)}{\log(\alpha^+/\beta^+)}c_1^- > 0$.
\end{proof}

\subsection{Proof of Lemma \ref{lem:eig-W} (Spectral Gap)}
\begin{proof}
Since Assumption \ref{assump:base} holds, then by \eqref{distri:a-same} and \eqref{distri:a-diff}, we have
\begin{align*}
	\E[\bA^+] = \frac{p^++q^+}{2}\bE + \frac{p^+-q^+}{2}\bx^*\bx^{*^\top}-p^+\bm{I},\quad \E[\bA^-] = \frac{p^-+q^-}{2}\bE + \frac{p^--q^-}{2}\bx^*\bx^{*^\top}-p^-\bm{I}.
\end{align*}
These, together with $\bW=\tbA-\rho \bE$, $\tbA=\bA^+-\xi \bA^-$, and \eqref{eq:E-rho}, give
\begin{align}\label{eq:E-W}
	\E[\bW] & = \E[\tbA] - \E[\rho]\bE = \E[\bA^+] - \xi \E[\bA^-] - \E[\rho]\bE \notag\\
	& = \frac{1}{2}\left((p^+-q^+)-\xi(p^--q^-)\right)\bx^*\bx^{*^\top} + \frac{1}{n}(p^+-\xi p^-)\bE - (p^+-\xi p^-)\bI. 
\end{align}
Let $\nu_1 \ge \nu_2 \ge \cdots \ge \nu_n$ be the eigenvalues of $\left((p^+-q^+)-\xi(p^--q^-)\right)\bx^*\bx^{*^\top}/2$. Using \eqref{p+-} and \eqref{q+-}, one can compute 
\begin{align}\label{eq:nu}
	\nu_1 = \frac{\log n}{2}\left((\alpha^+-\beta^+)-\xi(\alpha^--\beta^-)\right),\quad  \nu_i=0 \text{\ for\ } i=2,\dots,n. 
\end{align}
Let $\bm{\Delta}\coloneqq \bW - \left((p^+-q^+)-\xi(p^--q^-)\right)\bx^*\bx^{*^\top}/2$, then it follows from  Weyl's inequality that
\begin{align}\label{eq1:lem-eig-W}
	\left| |\lambda_i| - |\nu_i| \right| \leq |\lambda_i - \nu_i| \le \left\|\bm{\Delta}\right\|,\ i=1,2,\cdots,n.
\end{align}
Besides, note that
\begin{align}\label{eq:W}
	\bW = \E[\bW] + (\tbA - \E[\tbA]) - \left(\rho - \E[\rho]\right) \bE.
\end{align}
According to Lemma \ref{lem:rho}, Lemma \ref{lem:Delta}, and the union bound, it holds with probability at least $1-2n^{-3}-2n^{-\frac{1}{2(\alpha^++\beta^++1)}}-2n^{-\frac{1}{2(\alpha^- + \beta^-+1)}}$
\begin{align*}
	\|(\tbA - \E[\tbA]) - \left(\rho - \E[\rho]\right) \bE\| \le \| \tbA - \E[\tbA] \| + n\left| \rho - \E[\rho] \right| \le c_1\sqrt{\log n} + (1+\xi)\frac{\log n}{\sqrt{n}} \le 2c_1\sqrt{\log n},
\end{align*}
where the last inequality is due to $n \ge (1+\xi)^2\log n/c_1^2$ for sufficiently large $n$. This, together with \eqref{eq:E-W} and \eqref{eq:W}, implies
\begin{align}\label{eq2:lem-eig-W}
	\|\bm{\Delta}\| = \left\|\bW - \frac{1}{2}\left((p^+-q^+)-\xi(p^--q^-)\right)\bx^*\bx^{*^\top}\right\| \le 2c_1\sqrt{\log n} + 2|p^+-\xi p^-| \le 3c_1\sqrt{\log n}  
\end{align}
where the last inequality is due to \eqref{p+-} and \eqref{xi} for sufficiently large $n$. By \eqref{eq:nu}, \eqref{eq1:lem-eig-W}, and \eqref{eq2:lem-eig-W}, for all sufficiently large $n$, it holds with probability at least $1-n^{-\Omega(1)}$ that
\begin{align}
	& |\lambda_1| \geq |\nu_1| - \|\bm{\Delta}\| \geq \frac{c_0}{2}\log n - 3c_1\sqrt{\log n},  \label{eq:lambda1-abs}\\ 
	&|\lambda_i| \leq |\nu_i| + \|\bm{\Delta}\| \leq 3c_1\sqrt{\log n},\ \text{for}\ i=2,\dots,n. \label{eq:lambdai-abs}
\end{align}
Moreover, applying the Davis-Kahan theorem (see, e.g., \citet[Theorem 4.5.5]{vershynin2018high}) to $\bW$ yields
\begin{align*}
	\min_{\theta \in \{\pm 1\}} \left\|\bu_1 - \theta\frac{\bx^*}{\sqrt{n}} \right\| \le \frac{2\sqrt{2}\left\|\bW - \frac{1}{2}\left((p^+-q^+)-\xi(p^--q^-)\right)\bx^*\bx^{*^\top}\right\|}{|\nu_1|}
	\le  \frac{12\sqrt{2}c_1}{\left|(\alpha^+-\beta^+)-\xi(\alpha^--\beta^-)\right|\sqrt{\log n}},
\end{align*}
where the last inequality is due to \eqref{eq:nu} and \eqref{eq2:lem-eig-W}, and the last inequality holds for sufficiently large $n$ and for $i = 2,\dots,n$. Setting $c_2\coloneqq 12\sqrt{2}c_1/c_0$ gives \eqref{eq:u1-x*}.
\end{proof}

\subsection{Proof of Proposition \ref{prop-prop1}}
\begin{proof}
Let $\bW = \bU\bm{\Lambda}\bU^\top$ be the eigenvalue decomposition of of $\bW$, where $\bm{\Lambda} = \diag(\lambda_1,\dots,\lambda_n)$ with $\lambda_1 \ge \cdots \ge \lambda_n$ being the eigenvalues and $\bU=\begin{bmatrix}
	\bu_1 & \cdots & \bu_n
\end{bmatrix}$ are the associated eigenvectors. Suppose that the results in Lemma \ref{lem:eig-W} holds, which happens with probability at least $1-n^{-\Omega(1)}$. 
First, by the power iteration in Algorithm \ref{alg-pmgpm}, we have
\begin{align*}
	\bm{y}^{(t)} = \frac{\bm{W}\bm{y}^{(t-1)}}{\|\bm{W}\bm{y}^{(t-1)}\|_2} = \frac{\bm{W}^2\bm{y}^{(t-2)}}{\|\bm{W}^2\bm{y}^{(t-2)}\|_2} = \cdots = \frac{\bm{W}^t\bm{y}^{(0)}}{\|\bm{W}^t\bm{y}^{(0)}\|_2}.
\end{align*}
Let the eigendecomposition of $\bm{W}$ be $\bm{W}=\bm{U}\bm{\Lambda}\bm{U}^\top$ and $\bm{b}=\bm{U}^\top\bm{y}^{(0)}$, then
\begin{align*}
	&\langle\bm{u}_1,\bm{y}^{(t)}\rangle 
	= \frac{\bm{u}_1^\top\bm{W}^{t}\bm{y}^{(0)}}{\|\bm{W}^{t}\bm{y}^{(0)}\|_2} 
	= \frac{\bm{u}_1^\top\bm{U}\bm{\Lambda}^{t}\bm{U}^\top\bm{y}^{(0)}}{\|\bm{U}\bm{\Lambda}^{t}\bm{U}^\top\bm{y}^{(0)}\|_2}
	=	\frac{\bm{u}_1^\top\bm{U}\bm{\Lambda}^{t}\bm{b}}{\|\bm{\Lambda}^{t}\bm{b}\|_2}
	=
	\frac{\lambda_1^t b_1}{\sqrt{\sum_{i=1}^{n}\lambda_i^{2t}b_i^2}}.
\end{align*}
Besides, Cauchy-Schwartz inequality reads
\begin{align*}
	|\langle\bm{u}_1,\bm{y}^{(t)}\rangle| \leq \|\bm{u}_1\|_2\|\bm{y}^{(t)}\|_2=1.
\end{align*}
We then consider the following two cases:
\begin{itemize}
	\item[$1^\circ$] Assume that $b_1\geq0$. Since $\lambda_1>0$, then $\langle\bm{u}_1,\bm{y}^{(t)}\rangle\geq0$, i.e., $0\leq\langle\bm{u}_1,\bm{y}^{(t)}\rangle \leq 1$. By \citet[Lemma 4]{wang2020nearly}, for sufficiently large $n$, it holds with probability at least $1-2n^{-1/2}$ that
	\begin{align}\label{eq:bi/b} \sum_{i=2}^{n}\left(\frac{b_i}{b_1}\right)^2\leq\frac{n^2}{2}.
	\end{align}
	Then, for sufficiently large $n$, it holds with probability at least $1-n^{-\Omega(1)}$ that
	\begin{align*}
		\|\bm{y}^{(t)}-\bm{u}_1\|_2^2 
		&= \|\bm{y}^{(t)}\|_2^2 + \|\bm{u}_1\|_2^2 - 2\langle\bm{u}_1,\bm{y}^{(t)}\rangle = 2 \left(1- \langle\bm{u}_1,\bm{y}^{(t)}\rangle\right) \\
		&\leq 2 \left(1- \langle\bm{u}_1,\bm{y}^{(t)}\rangle^2\right)
		= 2 \frac{\sum_{i=2}^{n}\lambda_i^{2t}b_i^2}{\sum_{i=1}^{n}\lambda_i^{2t}b_i^2} 
		\leq 2 \frac{\sum_{i=2}^{n}\lambda_i^{2t}b_i^2}{\lambda_1^{2k}b_1^2}
		\leq 2 \left|\frac{\lambda'}{\lambda_1}\right|^{2t}\sum_{i=2}^{n}\left(\frac{b_i}{b_1}\right)^2\\
		&\leq \left( \frac{6c_1}{|\alpha^+-\beta^+-\xi (\alpha^--\beta^-)|\sqrt{\log n} - 6c_1} \right)^{2t} n^2,
	\end{align*}
	where $\lambda'\coloneqq\max\{|\lambda_2|,\dots,|\lambda_n|\}$ and the last inequality follows from \eqref{eq:lambda1-abs}, \eqref{eq:lambdai-abs}, and \eqref{eq:bi/b}. That is,
	\begin{align*}
		\|\bm{y}^{(t)}-\bm{u}_1\|_2
		&\leq n\left( \frac{6c_1}{|\alpha^+-\beta^+-\xi (\alpha^--\beta^-)|\sqrt{\log n} - 6c_1} \right)^{t}.
	\end{align*}
	\item[$2^\circ$] Assume that $b_1<0$. Since $\lambda_1>0$, then $\langle\bm{u}_1,\bm{y}^{(t)}\rangle<0$, i.e., $-1\leq\langle\bm{u}_1,\bm{y}^{(t)}\rangle < 0$. Following similar arguments with case $1^\circ$, we have
	\begin{align*}
		\|\bm{y}^{(t)}+\bm{u}_1\|_2^2 
		&\leq 2 \left(1- \langle\bm{u}_1,\bm{y}^{(t)}\rangle^2\right) 
		\leq \left( \frac{6c_1}{|\alpha^+-\beta^+-\xi (\alpha^--\beta^-)|\sqrt{\log n} - 6c_1} \right)^{2t} n^2,
	\end{align*}
	i.e.,
	\begin{align*}
		\|\bm{y}^{(t)}+\bm{u}_1\|_2
		&\leq n\left( \frac{6c_1}{|\alpha^+-\beta^+-\xi (\alpha^--\beta^-)|\sqrt{\log n} - 6c_1} \right)^{t}.
	\end{align*}
\end{itemize}
Combining cases $1^\circ$ and $2^\circ$, we have for all $t\geq 0$ and for sufficiently large $n$, it holds with probability at least $1-n^{-\Omega(1)}$ that
\begin{align*}
	\min_{\theta\in\{\pm 1\}}\ \|\bm{y}^{(t)}-\theta\bm{u}_1\|_2 \leq n\left( \frac{6c_1}{c_0\sqrt{\log n} - 6c_1} \right)^{t}.
\end{align*}
\end{proof}

\section{Proofs in Section \ref{subsec:GPM}}

\subsection{Proof of Lemma \ref{lem:Wx}}
\begin{proof}
Let $m \coloneqq n/2$ and
\begin{align*}
	& I_1 \coloneqq \left\{(i,j):\ \text{nodes}\ i,\ j\ \text{belong to the same community}\right\},\\
	& I_2 \coloneqq \left\{(i,j):\ \text{nodes}\ i,\ j\ \text{belong to different communities}\right\}.
\end{align*}
Then, it follows from $\bE\bx^*=\bm{0}$, \eqref{p+-}, \eqref{q+-}, and \eqref{eq:aij} that
\begin{align*}
	x_i^*(\bW\bm{x}^*)_i & = x_i^*(\tbA\bm{x}^*)_i = x_i^*\left((\bA^+-\xi\bA^-)\bx^*\right)_i \\ 
	& = \sum_{j:(i,j) \in I_1}(A_{ij}^+-\xi A_{ij}^-)x_i^*x_j^* + \sum_{j: (i,j) \in I_2}(A_{ij}^+-\xi A_{ij}^-)x_i^*x_j^* = \sum_{j=1}^{m-1} X_j - \sum_{j=1}^{m} Y_j,
\end{align*}
where $\left\{X_j:j=1,\dots,m-1 \right\}$  are i.i.d. random variables whose distributions are given by
\begin{equation}\label{dist:Xj}
	X_j = 
	\begin{cases}
		1,& \text{w.p.\ \ } p^+,\\
		-\xi,& \text{w.p.\ \ } p^-,\\
		0,& \text{w.p.\ \ } 1- (p^++p^-),
	\end{cases}
	\text{\ for\ } j = 1,\dots,m,
\end{equation}
and  $\left\{Y_j: j=1,\dots,m\right\}$, independent of $\left\{X_j: j=1,\dots,m-1\right\}$, are i.i.d. random variables whose distributions are given by  
\begin{equation}\label{dist:Yj}
	Y_j = 
	\begin{cases}
		1,& \text{w.p.\ \ } q^+,\\
		-\xi,& \text{w.p.\ \ } q^-,\\
		0,& \text{w.p.\ \ } 1- (q^++q^-),
	\end{cases} 
	\text{\ for\ } j = 1,\dots,m.
\end{equation}
For any $\gamma>0$, applying Markov's inequality to the moment generating function with $\lambda<0$ yields 
\begin{align*}
	\mathbb{P}\left(x_i^*(\bW\bm{x}^*)_i\leq\gamma\log n\right) & =  \mathbb{P}\left(\sum_{j=1}^{m-1}X_j - \sum_{j=1}^{m}Y_j\leq\gamma\log n\right)
	\leq \mathbb{P}\left(\sum_{j=1}^{m}X_j - \sum_{j=1}^{m}Y_j\leq\gamma\log n+1\right) \\
	& \leq \mathbb{P}\left(\sum_{j=1}^{m}X_j - \sum_{j=1}^{m}Y_j \leq \gamma'\log n\right) = \mathbb{P}\left(\exp\left(\lambda\left(\sum_{j=1}^{m}X_j - \sum_{j=1}^{m}Y_j\right)\right)\geq\exp(\lambda\gamma'\log n)\right) \\
	& \le n^{-\lambda\gamma'}  \mathbb{E}\left[\exp\left(\lambda\left(\sum_{j=1}^{m}X_j - \sum_{j=1}^{m}Y_j\right)\right)\right]  \\
	& = n^{-\lambda\gamma'} \mathbb{E}\left[\exp\left(\lambda\sum_{j=1}^{m}X_j \right)\right] \mathbb{E}\left[\exp\left(-\lambda\sum_{j=1}^{m}Y_j\right)\right] \\
	&= n^{-\lambda\gamma'} \prod_{j=1}^{m}\mathbb{E}\left[\exp(\lambda X_j)\right] \prod_{j=1}^{m}\mathbb{E}\left[\exp(-\lambda Y_j)\right] \\
	& = n^{-\lambda\gamma'} \left(p^+e^\lambda+p^-e^{-\xi\lambda}+1-(p^++p^-)\right)^{\frac{n}{2}} \left(q^+e^{-\lambda}+q^-e^{\xi\lambda}+1-(q^++q^-)\right)^{\frac{n}{2}},
\end{align*}
where the second inequality holds by $\gamma' \ge \gamma+1/\log n$, the third equality is due to the fact that $\left\{Y_j: j=1,\dots,m\right\}$ are independent of $\left\{X_j: j=1,\dots,m\right\}$, and the last equality follows from \eqref{dist:Xj} and \eqref{dist:Yj}. Taking logarithm of the above inequality and using \eqref{p+-} and \eqref{q+-} yield
\begin{align}
	&\log \mathbb{P}\left(x_i^*(\bm{W}\bm{x}^*)_i\leq\gamma \log n\right)  \nonumber\\
	&\leq -\lambda\gamma'\log n + \frac{n}{2} \log\left(\frac{\alpha^+\log n}{n}(e^\lambda-1)+\frac{\alpha^-\log n}{n}(e^{-\xi\lambda}-1)+1\right) \nonumber\\
	&\qquad + \frac{n}{2}\log\left(\frac{\beta^+\log n}{n}(e^{-\lambda}-1)+\frac{\beta^-\log n}{n}(e^{\xi\lambda}-1)+1\right) \nonumber\\
	&\leq -\lambda\gamma'\log n + \frac{n}{2} \left(\frac{\alpha^+\log n}{n}(e^\lambda-1)+\frac{\alpha^-\log n}{n}(e^{-\xi\lambda}-1)\right) \nonumber\\
	&\qquad + \frac{n}{2}\left(\frac{\beta^+\log n}{n}(e^{-\lambda}-1)+\frac{\beta^-\log n}{n}(e^{\xi\lambda}-1)\right) \nonumber\\
	&= \left(-\lambda\gamma'+\frac{\alpha^+}{2}(e^\lambda-1)+\frac{\alpha^-}{2}(e^{-\xi\lambda}-1)+\frac{\beta^+}{2}(e^{-\lambda}-1)+\frac{\beta^-}{2}(e^{\xi\lambda}-1)\right)\log n, \label{eq:logP}
\end{align}
where the second inequality is due to $\log(1+x) \le x$ for any $x > -1$. Taking $\lambda = - \log(\alpha^+/\beta^+)/2 <0$, we have
\begin{align*}
	e^\lambda =  \exp\left(  -\frac{ \log(\alpha^+/\beta^+)}{2} \right) = \sqrt{\frac{\beta^+}{\alpha^+}},\quad e^{\xi\lambda} = \exp\left(  -\frac{\log(\beta^-/\alpha^-)}{\log(\alpha^+/\beta^+)}\cdot \frac{ \log(\alpha^+/\beta^+)}{2} \right) = \sqrt{\frac{\alpha^-}{\beta^-}}.
\end{align*}		
It follows from this that
\begin{align*}
	& \frac{\alpha^+}{2}(e^\lambda-1)+\frac{\alpha^-}{2}(e^{\xi\lambda}-1)+\frac{\beta^+}{2}(e^{-\lambda}-1)+\frac{\beta^-}{2}(e^{-\xi\lambda}-1)\\
	&= \frac{\alpha^+}{2}\left( \sqrt{\frac{\beta^+}{\alpha^+}} - 1\right) +  \frac{\alpha^-}{2}\left(\sqrt{\frac{\beta^-}{\alpha^-}} - 1\right) + \frac{\beta^+}{2} 
	\left( \sqrt{\frac{\alpha^+}{\beta^+}} - 1\right) +  \frac{\beta^-}{2}\left(\sqrt{\frac{\alpha^-}{\beta^-}} - 1\right) \\
	&= \sqrt{\alpha^+\beta^+}-\frac{\alpha^+}{2}-\frac{\beta^+}{2} + \sqrt{\alpha^-\beta^-}-\frac{\alpha^-}{2}-\frac{\beta^-}{2} \\
	&= -\frac{1}{2} (\sqrt{\alpha^+}-\sqrt{\beta^+})^2-\frac{1}{2}(\sqrt{\alpha^-}-\sqrt{\beta^-})^2.
\end{align*}
This, together with \eqref{eq:logP}, yields 
\begin{align*}
	&\log \mathbb{P}\left(x_i^*(\bm{W}\bm{x}^*)_i\leq\gamma\log n\right)  \le  \frac{1}{2}\left(\gamma'\log\left(\frac{\alpha^+}{\beta^+}\right)- (\sqrt{\alpha^+}-\sqrt{\beta^+})^2-(\sqrt{\alpha^-}-\sqrt{\beta^-})^2\right) \log n
\end{align*}
Thus, we have  
\begin{align*}
	&\mathbb{P}\left(x_i^*(\bm{W}\bm{x}^*)_i\leq\gamma\log n\right) \leq n^{-\frac{1}{2}\left((\sqrt{\alpha^+}-\sqrt{\beta^+})^2+(\sqrt{\alpha^-}-\sqrt{\beta^-})^2\right)+\frac{\gamma'}{2}\log\left(\frac{\alpha^+}{\beta^+}\right)}.
\end{align*}
According to the union bound, it holds that
\begin{align*}
	\mathbb{P}\left(\min\{x_i^*(\bm{W}\bm{x}^*)_i:i=1,\dots,n\} \geq \gamma\log n\right) \ge 1- n^{1-\frac{1}{2}\left((\sqrt{\alpha^+}-\sqrt{\beta^+})^2+(\sqrt{\alpha^-}-\sqrt{\beta^-})^2\right)+\frac{\gamma'}{2}\log\left(\frac{\alpha^+}{\beta^+}\right)}.
\end{align*}
In particular, by $\alpha^+>\beta^+$ and \eqref{IT:bound}, for sufficiently large $n$, there exist constants $\gamma$ and $\gamma'$ satisfying
\begin{align*}
	0 < \gamma+\frac{1}{\log n} < \gamma' < \frac{(\sqrt{\alpha^+}-\sqrt{\beta^+})^2+(\sqrt{\alpha^-}-\sqrt{\beta^-})^2 - 2}{2\log(\alpha^+/\beta^+)} 
\end{align*}
such that
\begin{align*}
	1-\frac{1}{2}\left((\sqrt{\alpha^+}-\sqrt{\beta^+})^2+(\sqrt{\alpha^-}-\sqrt{\beta^-})^2\right)+\frac{\gamma'}{2}\log\left(\frac{\alpha^+}{\beta^+}\right) < 0.
\end{align*}
\end{proof}

\subsection{Proof of Proposition \ref{prop:contraction} (Contraction Property) and Its Corollary}
\begin{proof}	
We only prove statement (i), then statement (ii) can be shown similarly. According to Lemma \ref{lem:Wx}, Lemma \ref{lem:rho}, Lemma \ref{lem:Delta}, \eqref{rst:lem:Wx}, \eqref{rst:lem:rho}, and \eqref{rst:lem:Delta} hold with probability at least $1-n^{-\Omega(1)}$
Let $\bm{\widetilde{\Delta}}\coloneqq\tbA - \E[\tbA]$. According to \eqref{eq:E-W} and \eqref{eq:W}, we have
\begin{align}\label{eq1:prop:contr}
	\bm{W} = \frac{p^+-q^+-\xi(p^--q^-)}{2}\bm{x}^*\bm{x}^{*^\top} + \frac{1}{n}(p^+-\xi p^-)\bE - (p^+-\xi p^-)\bm{I}  + \bm{\widetilde{\Delta}} - \left(\rho - \E[\rho]\right) \bE.
\end{align}
Let $\bx \in \R^n$ be arbitrary such that $\|\bm{x}\| = \sqrt{n}$. Since $\|\bm{x}\| =\|\bm{x}^*\| =\sqrt{n}$, we have
\begin{align*}
	\|\bm{x}-\bm{x}^*\|^2 &= \|\bm{x}\|^2 + \|\bm{x}^*\|^2 - 2\bm{x}^{*^\top}\bm{x} =  2\bm{x}^{*^\top}\bm{x}^* - 2\bm{x}^{*^\top}\bm{x},
\end{align*}
which implies
\begin{align}
	\left|\bm{x}^{*T}(\bm{x}^*-\bm{x})\right| = \frac{1}{2}\|\bm{x}-\bm{x}^*\|^2. \label{eq:x-x*}
\end{align}
Then, we have
\begin{align}
	&\|\bm{W}\bm{x}-\bm{W}\bm{x}^*\| \notag \\
	& \le \frac{1}{2}\left|p^+-q^+-\xi(p^--q^-)\right|\|\bm{x}^*\bm{x}^{*^\top}(\bm{x}-\bm{x}^*)\| + \frac{1}{n}\|(p^+-\xi p^-)\bE(\bx-\bx^*)\|  \notag\\
	& \qquad\ +\left|p^+-\xi p^-\right|\left\|\bm{x}-\bm{x}^*\right\|  + \|\bm{\widetilde{\Delta}}(\bm{x}-\bm{x}^*)\| + \left\|\left(\rho - \E[\rho]\right) \bE(\bm{x}-\bm{x}^*)\right\| \notag\\
	& \leq \frac{\sqrt{n}}{4}\left|p^+-q^+-\xi(p^--q^-)\right|\|\bm{x}-\bm{x}^*\|^2 + 2\left|p^+-\xi p^-\right|\left\|\bm{x}-\bm{x}^*\right\| +
	\|\bm{\widetilde{\Delta}}\|\|\bm{x}-\bm{x}^*\| +  n\left|\rho - \E[\rho]\right|\left\|\bm{x}-\bm{x}^*\right\| \notag\\
	& \le \frac{c_0\log n}{4\sqrt{n}}\|\bm{x}-\bm{x}^*\|^2 + \left(\frac{2\log n}{n}\left|\alpha^+-\xi \alpha^-\right| + c_1\sqrt{\log n} + (1+\xi)\frac{\log n}{\sqrt{n}} \right)\left\|\bm{x}-\bm{x}^*\right\|,\label{eq2:prop:contr}
\end{align}
where the first inequality is due to \eqref{eq1:prop:contr} and the triangular inequality, the second inequality follows from \eqref{eq:x-x*}, and the last inequality follows from  \eqref{p+-}, \eqref{q+-}, Lemma \ref{lem:rho} and Lemma \ref{lem:Delta}.

By \eqref{rst:lem:Wx} in Lemma \ref{lem:Wx}, we have
\begin{align*}
	\frac{\bW\bx^*}{|\bW\bx^*|} = \sign(\bm{W}\bm{x}^*) = \sign(\bm{x}^*) = \bm{x}^*,
\end{align*}
and 
\begin{align}\label{eq3:prop:contr}
	\min\{|(\bm{W}\bm{x}^*)_i|: i=1,\dots,n\} \geq \gamma\log n,
\end{align}
Thus, for $\bx$ satisfying $\|\bm{x}\| =\sqrt{n}$ and 
\begin{align}
	\|\bm{x}-\bm{x}^*\| \leq  2c_2\sqrt{\frac{n}{\log n}}
	= \frac{24\sqrt{2}c_1}{c_0}\sqrt{\frac{n}{\log n}} \label{x:init1}
\end{align} 
we have 
\begin{align*}
	\left\| \frac{\bW\bx}{|\bW\bx|} -\bm{x}^*\right\| 
	=& \left \|\frac{\bW\bx}{|\bW\bx|} - \frac{\bW\bx^*}{|\bW\bx^*|} \right\|  
	\leq  \frac{2\|\bm{W}\bm{x}-\bm{W}\bm{x}^*\|}{\gamma\log n} \\
	\leq& \frac{c_0}{2\gamma\sqrt{n}}  \|\bm{x}-\bm{x}^*\|^2 + \left(\frac{4}{\gamma n}\left|\alpha^+-\xi \alpha^-\right| + \frac{2c_1}{\gamma\sqrt{\log n}} + \frac{2(1+\xi)}{\gamma\sqrt{n}} \right)\left\|\bm{x}-\bm{x}^*\right\| \\
	\leq& \left(\frac{c_0}{2\gamma\sqrt{n}}\|\bm{x}-\bm{x}^*\| + \frac{4|\alpha^+-\xi \alpha^-|+2(1+\xi)}{\gamma\sqrt{n}}+\frac{2c_1}{\gamma\sqrt{\log n}} \right)\left\|\bm{x}-\bm{x}^*\right\|\\
	\leq& \left(\frac{12\sqrt{2}c_1}{\gamma\sqrt{\log n}} + \frac{4|\alpha^+-\xi \alpha^-|+2(1+\xi)}{\gamma\sqrt{\log n}} + \frac{2c_1}{\gamma\sqrt{\log n}}\right)\left\|\bm{x}-\bm{x}^*\right\|  \\
	=& \frac{\left(12\sqrt{2}+2\right)c_1+4|\alpha^+-\xi \alpha^-|+2(1+\xi)}{\gamma\sqrt{\log n}} \left\|\bm{x}-\bm{x}^*\right\|,
\end{align*}
where the first inequality is due to \citet[Lemma 6]{wang2020nearly} and \eqref{eq3:prop:contr}, the second inequality follows from \eqref{eq2:prop:contr}, and the fourth inequality uses \eqref{x:init1}. Setting $c_3 = (6\sqrt{2}+1)c_1+2|\alpha^+-\xi \alpha^-|+ 1+\xi$ gives the desired result. 
\end{proof}
Given Proposition \ref{prop:contraction}, we immediately have the following corollary which specifies the linear convergence of the sequence $\{\bm{x}^t\}_{t\geq 0}$ in the second stage of Algorithm \ref{alg-pmgpm} within the contraction region. 
\begin{corollary}\label{coro:coro1}
Suppose that Assumption \ref{assump:base} holds and $\alpha^+,\beta^+,,\beta^-,\alpha^-$ satisfy \eqref{IT:bound}. Let $\{\bm{x}\}_{t\geq 0}$ be the sequence generated by the second stage of Algorithm \ref{alg-pmgpm}. Then, for sufficiently large $n$, either of the following two statements holds with probability at least $1-n^{-\Omega(1)}$: \\
(i) If $\|\bm{x}^0-\bm{x}^*\| \leq \frac{8\sqrt{2}c_1}{c_0}\sqrt{{n}/{\log n}}$, then for all $t\geq 0$, we have
\begin{align*}
	\|\bm{x}^t-\bm{x}^*\|   
	\leq  \left(\frac{2c_3}{\gamma\sqrt{\log n}}\right)^t\|\bm{x}^0-\bm{x}^*\|;
\end{align*}
(ii) If $\|\bm{x}^0+\bm{x}^*\| \leq \frac{8\sqrt{2}c_1}{c_0}\sqrt{{n}/{\log n}}$, then for all $t\geq 0$, we have
\begin{align*}
	\|\bm{x}^t+\bm{x}^*\|   
	\leq \left(\frac{2c_3}{\gamma\sqrt{\log n}}\right)^t\|\bm{x}^0+\bm{x}^*\|,
\end{align*}
where $c_0$ and $c_1$ are the constants in Lemma \ref{lem:eig-W}, $c_3$ is the constant in Proposition \ref{prop:contraction}, and $\gamma>0$ is the constant in Lemma \ref{lem:Wx}.
\end{corollary}

\subsection{Proof of Proposition \ref{prop-one-step-conv} (One-Step Convergence)}
\begin{proof}
We only prove statement (i), then statement (ii) can be shown similarly. Since $\bm{x}\in\{1,-1\}^n$ and $\|\bm{x}-\bm{x}^*\|_2\leq2$, then there exists $\ell\in[n]$, such that $\bm{x}=\bm{x}^*$ or $\bm{x}=\bm{x}^*\pm2\bm{e}_\ell$, where $\bm{e}_\ell$ is an $n$-dimensional vector with $(\bm{e}_\ell)_\ell=1$ and $(\bm{e}_\ell)_i=0$ for all $i\neq\ell$. Hence, we consider the following two cases:
\begin{itemize}
	\item[$1^\circ$] If $\bm{x}=\bm{x}^*$, then by \eqref{rst:lem:Wx} in Lemma \ref{lem:Wx}, it holds with probability at least $1-n^{-\Omega(1)}$ that
	\begin{align*}
		x_i^*(\bm{W}\bm{x}^*)_i > \gamma\log n > 0,
	\end{align*}
	for all $i=1,\dots,n$.
	\item[$2^\circ$] If $\bm{x}=\bm{x}^*\pm2\bm{e}_\ell$, then by Lemma \ref{lem:Wx}, it holds with probability at least $1-n^{-\Omega(1)}$ that,
	\begin{align*}
		{x}^*_i(\bm{W}\bm{x})_i 
		&= {x}^*_i(\bm{W}\bm{x}^*\pm2\bm{W}\bm{e}_\ell)_i 
		= {x}^*_i(\bm{W}\bm{x}^*)_i \pm 2{x}^*_i(\bm{W}\bm{e}_\ell)_i 
		\geq \gamma\log n \pm 2{x}^*_i(\bm{W}\bm{e}_\ell)_i,
	\end{align*}
	for all $i=1,\dots,n$. 
	Note that $\bm{W}\bm{e}_\ell = \left(\widetilde{\bm{A}}-\rho\bm{E}\right)\bm{e}_\ell$, then for $i=1,\dots,n$, we have $(\bm{W}\bm{e}_\ell)_i = \widetilde{A}_{i\ell} - \rho$. By Lemma \ref{lem:rho}, we have
	\begin{align*}
		\left| \widetilde{A}_{i\ell} - (\bm{W}\bm{e}_\ell)_i - \frac{p^++q^+-\xi(p^-+q^-)}{2} + \frac{p^+- \xi p^-}{n} \right|
		\leq \frac{\log n}{n^{3/2}},
	\end{align*}
	which implies that
	\begin{align*}
		(\bm{W}\bm{e}_\ell)_i &\geq \widetilde{A}_{i\ell} - \frac{p^++q^+-\xi (p^-+q^-)}{2} + \frac{p^+- \xi p^-}{n} - \frac{\log n}{n^{3/2}}\\
		&\geq  -\xi - \frac{(\alpha^++\beta^+-\xi (\alpha^-+\beta^-))\log n}{n} - \frac{\log n}{n^{3/2}},
	\end{align*}
	and
	\begin{align*}
		-(\bm{W}\bm{e}_\ell)_i &\geq -\widetilde{A}_{i\ell} + \frac{p^++q^+-\xi (p^-+q^-)}{2} - \frac{p^+- \xi p^-}{n} + \frac{\log n}{n^{3/2}}\\
		&\geq -1 + \frac{(\alpha^++\beta^+-\xi (\alpha^-+\beta^-))\log n}{N} - \frac{p^+- \xi p^-}{n} + \frac{\log n}{n^{3/2}}
	\end{align*}
	Therefore, for sufficiently large $n$, it holds with probability at least $1-n^{-\Omega(1)}$ that
	\begin{align*}
		x_i^*(\bm{W}\bm{x})_i \geq \gamma\log n \pm 2{x}^*_i(\bm{W}\bm{e}_\ell)_i > 0,
	\end{align*}
	for all $i=1,\dots,n$.
\end{itemize}
Combining $1^\circ$ and $2^\circ$, we obtain $x_i^*(\bm{W}\bm{x})_i > 0$ for all $i=\{1,\dots,n\}$, which implies that $$\frac{\bW\bx}{|\bW\bx|} = \bx^*.$$
\end{proof}

\section{Proof of Theorem \ref{thm-iter-comp} and Corollary \ref{coro-coro2}} \label{proof-thm-coro}
We first prove Theorem \ref{thm-iter-comp}.
\begin{proof}
Suppose with out loss of generality that the following two inequalities hold in Propositions \ref{prop-prop1} and Corollary \ref{coro:coro1} simultaneously with probability at least $1-n^{-\Omega(1)}$:
\begin{align}
	\left\|\bm{y}^t - \bm{u}_1\right\| &\le n\left( \frac{6c_1}{c_0\sqrt{\log n} - 6c_1} \right)^t, \label{eq-yt-u1}\\
	\|\bm{x}^t-\bm{x}^*\|   
	&\leq  \left(\frac{2c_3}{\gamma\sqrt{\log n}}\right)^t\|\bm{x}^0-\bm{x}^*\|,\label{eq-corollary}
\end{align}
for $\bm{x}^0$ satisfying $\|\bm{x}^0\|_2=\sqrt{n}$ and
\begin{align}\label{x:init2}
	\|\bm{x}^0-\bm{x}^*\| \leq \frac{2c_2\sqrt{n}}{\sqrt{\log n}}.
\end{align}
Suppose that Algorithm \ref{alg-pmgpm} runs $T_1$ iterations in the first stage and outputs $\bm{y}\in\mathbb{R}^n$ such that $\|\bm{y}\|_2=1$ and 
\begin{equation}\label{eq-yT1-u1}
	\|\bm{y}^{(T_1)}-\bm{u}_1\|_2 \leq \frac{c_2}{\sqrt{\log n}}.
\end{equation}
By \eqref{eq-yt-u1}, it suffices to let
\begin{equation*}
	n\left( \frac{6c_1}{c_0\sqrt{\log n} - 6c_1} \right)^{T_1} \leq \frac{c_2}{\sqrt{\log n}},
\end{equation*}
which implies that $T_1 = \mathcal{O}(\log n/\log\log n)$. Then, the second stage of Algorithm \ref{alg-pmgpm} starts with $\bm{x}^0=\sqrt{n}\bm{y}^{(T_1)}$, thus by \eqref{eq-yT1-u1} and Lemma \ref{lem:eig-W}, we have 
\begin{align*}
	\|\bm{x}^0-\bm{x}^*\|_2 &\leq \|\bm{x}^0-\sqrt{n}\bm{u}_1\|_2 + \|\sqrt{n}\bm{u}_1-\bm{x}^*\|_2 
	= \sqrt{n}\|\bm{y}^{(T_1)}-\bm{u}_1\|_2 + \sqrt{n}\left\|\bm{u}_1-\frac{\bm{x}^*}{\sqrt{n}}\right\|_2 \\
	&\leq \frac{c_2\sqrt{n}}{\sqrt{\log n}} + \frac{c_2\sqrt{n}}{\sqrt{\log n}} 
	= 2c_2\sqrt{\frac{n}{\log n}}.
\end{align*}
Suppose that the second stage of Algorithm \ref{alg-pmgpm} terminates at $\bm{x}^*$ or $-\bm{x}^*$ with $T_2$ iterations. By \eqref{eq-corollary}, \eqref{x:init2}, and Proposition \ref{prop-one-step-conv}, we have
\begin{align*}
	\left(\frac{2c_3}{\gamma\sqrt{\log n}}\right)^{T_2} \frac{2c_2\sqrt{n}}{\sqrt{\log n}} \leq 2.
\end{align*}
which implies that $T_2 = \mathcal{O}(\log n/\log\log n)$.
\end{proof}
To proceed, we require the following lemma.
\begin{lemma}[{\citet[Lemma 10]{wang2020nearly} }]\label{lemma-per-iter-comp}
It holds with probability at least $1-n^{-\Omega(1)}$ that the number of non-zero entries in $\bm{A}$ is less than $2(\alpha^++\alpha^-+\beta^++\beta^-)n\log n$.
\end{lemma}
Finally, we prove Corollary \ref{coro-coro2}.
\begin{proof}
By Theorem \ref{thm-iter-comp}, the overall iteration complexity is $\mathcal{O}(\log n)$. By Lemma \ref{lemma-per-iter-comp}, the time complexity of each iteration is $\mathcal{O}(n\log n)$. Hence, the overall time complexity is $\mathcal{O}(n\log^2 n/\log\log n)$.
\end{proof}

\end{document}